\RequirePackage{etex}
\documentclass[5p,times]{elsarticle}

\usepackage[utf8]{inputenc}
\usepackage{xspace}
\usepackage{lineno,hyperref}
\modulolinenumbers[5]
\usepackage{hyperref}
\hypersetup{
  pdftitle={PRoVeFL: Private Robust and Verifiable Aggregation in Federated Learning},
  pdfauthor={Harsh Kasyap, Anil Kumar Pradhan, Ugur Ilker Atmaca, Graham Cormode, Carsten Maple}}
\usepackage{comment}
\usepackage{tikz}
\usepackage{tikz-network}
\usetikzlibrary{arrows.meta}
\usetikzlibrary{positioning,fit,backgrounds}
\usetikzlibrary{decorations.markings}
\usepackage{pgfplots}
\usepackage{pgfplotstable}
\pgfplotsset{compat=1.3}
\usepackage{multirow}
\usepackage{xcolor}
\usepackage{caption}
\usepackage{float}
\usepackage{subcaption}
\usepackage[section]{placeins}
\usepackage{graphicx}
\usepackage[ruled,vlined]{algorithm2e}
\usepackage{amsmath,amsfonts,amsthm}
\usepackage{algpseudocode}
\usepackage[normalem]{ulem}
\usepackage{enumitem}
\usepackage{pifont}
\newcommand{\cmark}{\ding{51}}%
\newcommand{\xmark}{\ding{55}}%

\newcommand{\sysname}{\textsc{PRoVeFL}\xspace}

\usepackage[skins]{tcolorbox}
\newtcolorbox{mybox}[2][]{%
  attach boxed title to top center
               = {yshift=-11pt},
  colframe     =black,
  colbacktitle = black,
  title        = #2,#1,
  enhanced,
}

\usepackage{array}
\newcolumntype{L}[1]{>{\raggedright\let\newline\\\arraybackslash\hspace{0pt}}m{#1}}
\newcolumntype{C}[1]{>{\centering\let\newline\\\arraybackslash\hspace{0pt}}m{#1}}
\newcolumntype{R}[1]{>{\raggedleft\let\newline\\\arraybackslash\hspace{0pt}}m{#1}}

\newtheorem{theorem}{Theorem}

\makeatletter
\renewcommand\paragraph{\@startsection{paragraph}{4}{\z@}%
  {1.25ex \@plus 1ex \@minus .2ex}%
  {-1em}%
  {\normalfont\normalsize\itshape}}
\makeatother

\DeclareMathOperator{\Encop}{Enc}
\newcommand{\Enc}[1]{\Encop_{\mathsf{pk}}\left(#1\right)}

\journal{arXiv}

\begin{document}
\begin{frontmatter}

\title{\sysname: Private Robust and Verifiable Aggregation in Federated Learning}

\author[label1,label3]{Harsh Kasyap\corref{cor}}
\ead{hkasyap.cse@iitbhu.ac.in}

\author[label2]{Anil Kumar Pradhan\corref{cor}}
\ead{anil@vaulttree.com}

\author[label3,label5]{Ugur Ilker Atmaca\corref{cor}}
\ead{ugur-ilker.atmaca@warwick.ac.uk}

\author[label4]{Graham Cormode}
\ead{graham.cormode@cs.ox.ac.uk}

\author[label3]{Carsten Maple}
\ead{cm@warwick.ac.uk}

\affiliation[label1]{organization={Indian Institute of Technology (BHU)}, city={Varanasi}, country={India}}
\affiliation[label2]{organization={Vaulttree}, city={}, country={Ireland}}
\affiliation[label3]{organization={University of Warwick}, city={Coventry}, country={UK}}
\affiliation[label4]{organization={University of Oxford}, city={Oxford}, country={UK}}
\affiliation[label5]{organization={Abdullah Gul University}, city={Kayseri}, country={Turkiye}}

\cortext[cor]{Corresponding author. First three authors have equal contribution.}
%\fntext[eqc]{First three authors have equal contribution.}


% compsoc conferences compared to traditional conferences), use this
%Telephone: (800) 555--1212, Fax: (888) 555--1212}


%\maketitle

\begin{abstract}
Federated Learning (FL) enables multiple clients to collaboratively train machine learning models while retaining data locality, thereby enhancing user privacy. However, traditional FL frameworks rely on a centralized aggregation server and assume honest-but-curious clients, making them susceptible to both server-side inference and client-side poisoning attacks. Although recent work has explored secure and Byzantine-resilient FL protocols, they face a fundamental trade-off among privacy, integrity, and verifiability, and incur substantial computational and communication overhead due to the heavy use of cryptographic primitives.

In this work, we propose \sysname--a novel, modular FL framework that is \underline{P}rivacy-preserving, Byzantine-\underline{Ro}bust, and ensures \underline{Ve}rifiable aggregation. \sysname employs multiple servers leveraging multi-key fully homomorphic encryption. Each client encrypts its local model updates and distributes encrypted shares to all servers. This design enables a hybrid computation model in which ciphertext operations are carefully offloaded to the plaintext domain under strict privacy constraints to efficiently evaluate complex statistical aggregation rules. \sysname is compatible with a wide range of state-of-the-art Byzantine-robust aggregation algorithms (e.g., Krum, Trimmed Mean, FLTrust, norm clipping, MESAS, and more) and further enhances them with verifiability mechanisms that require minimal trust in at least one honest server. We evaluate it across different settings and demonstrate its scalability with varying numbers of parameters and participants. \sysname improves runtime over the prior works, Prio and ELSA, based on distributed trust with comparable security guarantees, up to 100$\times$ and 10$\times$, respectively.
\end{abstract}

\begin{keyword}
Federated learning, Byzantine-robust aggregation, Multi-key homomorphic encryption, Verifiable aggregation, Privacy-preserving machine learning
\end{keyword}
\end{frontmatter}


% \fi
%\IEEEpeerreviewmaketitle

\section{Introduction}
%Financial fraud detection is an escalating challenge, with over \$3.1 trillion in illegitimate funds circulating globally in 2024 as a result of activities including fraud, money laundering, drug trafficking, and human exploitation \cite{nasdaq2024financialcrime}. Detecting and preventing these activities demands advanced machine learning models trained on large-scale, high-dimensional transaction data \cite{jiang2024spade+}. However, this data is often fragmented across institutions, each bound by strict regulatory constraints and business competition \cite{tian2025towards}. Furthermore, it is more challenging to detect novel fraudulent patterns. Federated Learning (FL) has emerged as a promising concept, enabling multiple financial institutions to collaboratively train shared models while maintaining the locality of raw data \cite{reddy2024deep, khan2024fed}. 

Federated Learning (FL) allows the development of a privacy-preserving distributed machine learning (ML) model for multiple clients to jointly train, without directly sharing their local data~\cite{mcmahan2016federated,mcmahan2017communication}. By decentralizing model training, FL offers significant advantages for sensitive applications such as healthcare and finance, where data is often fragmented between institutions, each bound by strict regulatory constraints and business competition \cite{tian2025towards}. In its standard design, a single server maintains a global ML model and incrementally updates it by aggregating local gradient updates sent by multiple clients. However, such a design is open to security and privacy threats~\cite{mothukuri2021survey}; in particular, an adversary can exploit model updates to \textit{infer} private information~\cite{nasr2019comprehensive,manna2022milsa} or \textit{manipulate} the aggregation process if the server is compromised~\cite{fu2022label}, and the global model can be \textit{poisoned} if clients are compromised~\cite{fang2020local,tolpegin2020data,kasyap2024sine}.

% \begin{figure}
%     \centering
%     \includegraphics[width=0.5\textwidth]{flissues.png}
%     \caption{Security issues in FL}
%     \label{fig:secissues}
% \end{figure}

To address such threats, three interrelated properties should be achieved together: privacy, Byzantine robustness, and verifiability \cite{rathee2023elsa}. %In a typical workflow, the clients locally compute their gradient updates, and the server maintains a global model by aggregating the gradient updates sent by the clients. However, 
Just hiding raw data is not sufficient for privacy.
As several studies~\cite{fu2022label,hu2021source} have confirmed, gradient updates can leak private information. Thus, other privacy-enhancing technologies (PETs), such as Secure Multi-Party Computation (SMPC), Fully Homomorphic Encryption (FHE), and Differential Privacy (DP), are integrated to enhance privacy.
While such protocols can hide clients' updates from the server, they can also make it harder to identify malicious updates. %Many studies addressed either clients' privacy, verifiability, or Byzantine robustness in isolation, but not all together. 
Recent work has proposed to combine secure aggregation with robustness measures, yet these often incur high computational and communication costs or rely on relaxed trust assumptions. 
They are not flexible enough to support a variety of Byzantine-robust aggregation rules, and instead enforce simple heuristics, such as via norm clipping~\cite{rathee2023elsa}. 
Other approaches, such as MUDGUARD~\cite{wang2024mudguard}, employ a privacy-preserving clustering method.
This can protect the global model against most malicious clients, but it significantly increases the complexity and cost of aggregation. 
Similarly, ELSA~\cite{rathee2023elsa}, ACORN~\cite{bell2023acorn}, and EIFFeL~\cite{roy2022eiffel} integrate zero-knowledge proofs or integrity checks into the secure aggregation to mitigate the impact of malicious clients by enforcing specified bounds or norms on client updates. However, they do not align with the current round distribution of local model updates; instead, they treat the previous global model update as the trusted reference.
%Moreover, they are limited to support operations such as norm clipping, and not able to integrate the state-of-the-art Byzantine-aggregation rules such as Krum~\cite{blanchard2017machine}, Trimmed-mean~\cite{yin2018byzantine}, FLTrust~\cite{cao2020fltrust} and others. %Also, they do not enable clients to verify the server's actions. %and often assume at least a semi-honest server in their threat models.
%Moreover, the ability to confirm that the central server aggregates updates correctly and without tampering further increases the privacy-robustness trade-off. Thus, achieving one of these properties may inadvertently decrease another.
It has also been shown that secure aggregations performed in vanilla FL settings are not secure, and can still leak sensitive and private information~\cite{gao2021secure}. 

The third critical property is the verifiability of the aggregation process. 
In standard FL, clients are often assumed to trust that the central server honestly performs the intended aggregation. However, a malicious or compromised server could deviate from the protocol by injecting maliciously crafted updates, omitting certain client contributions, or altering the aggregated result. Thus, it is important to verify that the server’s reported global model update truly corresponds to the \textit{correct aggregation, while correctly penalizing the malicious updates}. Protocols focusing on verifiable aggregation (\textit{e.g.,} VerifyNet~\cite{xu2019verifynet} or SVeriFL~\cite{gao2023sverifl}) enable clients to detect server misbehaviors by requiring the server to produce cryptographic proofs of aggregation. 
However, they assume that all client updates aggregated are legitimate. %Thus, there is an inherent trade-off for achieving privacy, Byzantine robustness and verifiability. 
It is a significant challenge to simultaneously achieve all three properties---client update privacy, Byzantine robustness, and server verifiability. 
Existing methods designed to achieve any of these obstruct or contradict the others.

%Figure~\ref{fig:secissues} highlights the tradeoff in order to achieve these three properties. It is evident that 
Fully Homomorphic Encryption (FHE)~\cite{gentry2009fully,cheon2017homomorphic} holds promise for achieving the key security properties by enabling computation directly over encrypted data and providing strong privacy guarantees without the need for trusted third parties. Furthermore, FHE can support Byzantine robustness and verifiability by enabling secure, auditable computation pipelines. Recent works such as FedML-HE~\cite{jin2023fedml}, xMK-CKKS~\cite{ma2022privacy}, and tMK-CKKS~\cite{du2023efficient} have demonstrated these capabilities to some extent. However, these systems still suffer from significant computational overhead and are not ready for adoption.
In particular, the non-linear, comparison-based logic used in Byzantine-resilient rules remains expensive to implement. Ensuring verifiability within a purely FHE-based framework often requires expensive zero-knowledge proofs or complex cryptographic commitments, further increasing costs. These limitations hinder scalability, make the protocols impractical for real-time systems, and restrict their applicability.
%\smallskip

\subsection{Our Contribution.} 
This paper proposes a generic \underline{P}rivate, \underline{Ro}bust, and \underline{Ve}rifiable FL (\sysname) framework to address the above challenges and ensure clients' privacy, Byzantine-robustness, and aggregation verifiability simultaneously. 
%We propose a novel, generic federated learning (FL) framework that introduces a 
\sysname is scalable with a modular multi-server architecture to address the computational and communication inefficiencies inherent in secure aggregation protocols. Unlike a conventional single-server FL system, \sysname partitions the model parameters (1) to distribute the computational load and (2) to enable offloading portions of the ciphertext-domain computation to the plaintext domain, without compromising privacy. This partitioning enables threshold decryption on a small number of servers, which would be infeasible with a large number of clients. By aligning the partitioning with CKKS-specific parameters such as the polynomial modulus degree, we maximize ciphertext packing efficiency, reduce the need for zero-padding, and thereby minimize computational and communication costs. %Assuming at least one honest server, which is essential to maintain correctness and security under our protocol.

\sysname ensures that client updates remain private from the servers unless all of them collude. Using multi-key homomorphic encryption, the protocol guarantees privacy as long as at least one server is honest, extending the threat model considered in ELSA \cite{rathee2023elsa}, and SafeFL \cite{gehlhar2023safefl}. It further extends the verification guarantees of MUDGUARD~\cite{wang2024mudguard}, and Franzese et al. \cite{franzese2023robust} by verifying the correctness of the aggregation and the server’s enforcement of the robustness rules. \sysname improves the scalability of the peer-to-peer approach used in Franzese et al.~\cite{franzese2023robust}, and its design ensures that security and performance scale with the number of servers without burdening clients with additional communication or computation. A detailed comparison of \sysname with existing FL schemes across key security properties is presented in Table~\ref{tab:study_comparison}.

\sysname is compatible with state-of-the-art Byzantine resilient aggregation rules.
We demonstrate this by implementing Krum~\cite{blanchard2017machine}, Trimmed Mean~\cite{yin2018byzantine}, FLTrust~\cite{cao2020fltrust} and MESAS~\cite{krauss2023mesas}.
This is a significant improvement over existing approaches that rely solely on relaxed rules, such as norm clipping. We also carefully choose these aggregation rules to cover complex operations such as (coordinate-wise) sorting, pairwise comparison, and cosine similarity. 
A central technical idea underlying \sysname is the use of a jointly generated random multiplicative mask \textit{r}, which enables a carefully controlled transition from encrypted to plaintext computation without compromising privacy. Specifically, each client encrypts its update under the group public key and homomorphically applies the encrypted mask $\text{Enc}_{pk}(r)$ to its local update shares $u_{ik}$, yielding $\text{Enc}_{pk}(r \cdot u_{ik})$. The servers then perform the necessary statistical operations to obtain $\text{Enc}_{pk}(r \cdot s_{ik})$, and then collaboratively decrypt the results. Crucially, decryption reveals only the masked intermediate values. This allows the computationally expensive selection and filtering steps to be completed in plaintext, at a fraction of the cost of performing them homomorphically.
%A central innovation in our design lies in the carefully structured transformation of critical aggregation steps from the encrypted to the plaintext domain, orchestrated in a manner that preserves both the additional computational requirements and end-to-end privacy guarantees. This selective offloading not only reduces the cost of inherently expensive ciphertext operations but also enables the integration of an efficient verifiability mechanism, provided that at least one server is honest. 
In summary, our main contributions are as follows:
\begin{itemize}
    \item \textbf{Privacy and Robustness:} \sysname is a generic and flexible multi-server FL framework that can integrate state-of-the-art Byzantine-robust aggregation schemes, while preserving the privacy of clients’ updates, utilizing multi-key homomorphic encryption. 

    \item \textbf{Efficiency:} \sysname improves efficiency and practicality by incorporating partial plaintext computation in the multi-server setting (by obfuscating ciphertexts to preserve order and privacy) and significantly reducing computation and communication costs to achieve scalability.

    \item \textbf{Verifiability:} \sysname provides verifiability for the entire aggregation process, including the guarantee of correct execution for Byzantine-robustness computations, with a minimal trust requirement of at least one honest server.

    \item \textbf{Evaluation:} We evaluate \sysname in different settings, with varying numbers of parameters, participants, and servers, while using state-of-the-art Byzantine-robust aggregation rules such as $L_2$ norm, Krum, Trimmed-Mean, and FLTrust. The results demonstrate the practicality and efficiency of \sysname.

    %\item %It propose a generic framework which allows to integrate Byzantine-robustness frameworks with preserving clients' privacy. (See elsa limitation)
    %\item %It does not only provide verification only for aggregation, but it provide verification for the Byzantine-robustness calculations has been done correctly with the required for only 1 honest server. (e.g., outlier filtering and coordinate-wise trimming) 
    %\item %The proposed scheme split the vector into parts between the server, that the increasing number of server helps to fasten the computation. And to ease the computation of cost HE, the proposed scheme converts the ciphertext to plaintext for computation without theoretically proven privacy guarantee.
\end{itemize}

\begin{table}[tb]
    \centering
    \caption{Qualitative comparison of related FL frameworks.}
    \resizebox{0.5\textwidth}{!}{%
    \begin{tabular}{l|c|c|c|c|c|c}
    \hline
    \textbf{Study} & \textbf{No. Serv.} & \multicolumn{3}{c|}{\textbf{Properties}} & \multicolumn{2}{c}{\textbf{Threat Model}}\\
    \cline{3-7}
    & & Priv. & Robust. & Verif. & Mal. Serv. & Mal. Client\\
    \hline
    SeaFlame \cite{tang2025seaflame} & 2  & \cmark  & \cmark & \xmark  & \cmark & \cmark \\
    \hline
    MUDGUARD \cite{wang2024mudguard} & $\ge 2$ &\cmark &\cmark &\cmark &\cmark  &\cmark \\
    \hline    
    SafeFL \cite{gehlhar2023safefl} & $\ge 2$ & \cmark & \cmark & \xmark & \cmark & \cmark \\
    \hline
    TAPFed \cite{xu2024tapfed} & $>2$ & \cmark & \xmark & \cmark & \cmark & \xmark \\
    \hline
    DefendFL \cite{liu2024defendfl} & 1 & \cmark &\cmark & \xmark & \xmark & \cmark \\
    %\hline
    %Hao et al. \cite{hao2023robust} &  &  &  &  &  &  \\
    \hline
    Franzese et al. \cite{franzese2023robust} & -- & \cmark & \cmark & \cmark & \cmark & \cmark \\
    \hline
    ELSA \cite{rathee2023elsa} &2 &\cmark  &\cmark &\xmark &\cmark &\cmark \\
    \hline
    Ma et al. \cite{ma2024trusted}& 2 & \cmark  & \xmark  &\cmark  &\cmark  &\cmark \\
    \hline
    ACORN \cite{bell2023acorn}&1  & \cmark &\cmark  & \xmark  &\cmark &\cmark \\
    \hline
    Tang et al. \cite{tang2024flexible}& 2 &\cmark  &\cmark & \xmark &\cmark &\cmark \\
    \hline
    Eiffel\cite{roy2022eiffel} &1 &\cmark &\cmark & \xmark & \xmark & \cmark\\
    \hline
    RoFL \cite{burkhalter2021rofl} &1& \cmark& \cmark & \xmark & \xmark &\cmark\\ 
    \hline
    Prio\cite{corrigan2017prio} &$\ge 2$ &\cmark & \cmark & \xmark & \xmark &\cmark \\
    \hline
    Prio+\cite{addanki2022prio+} &$\ge 2$ & \cmark &\cmark & \xmark & \xmark & \cmark \\
    \hline
    PVFL~\cite{yin2024pvfl} &$1$  &\cmark  &\xmark  & \cmark & \xmark &\xmark\\
    \hline
    VPPFL~\cite{huang2024vppfl} &$2$  &\cmark  &\cmark  & \cmark & \xmark &\cmark\\
    \hline
    \textbf{\sysname} &$\ge 2$  &\cmark  &\cmark\cmark$*$  & \cmark & \cmark &\cmark\\
    \hline
    \end{tabular}}
    \vspace{1mm}
    {\footnotesize $*$ \cmark\cmark refers to wide coverage of state-of-the-art robustness measures. }
    \label{tab:study_comparison}
\end{table}

\section{Background and Related Work}
In this section, we discuss the background of federated learning (FL) and its susceptibility to poisoning, inference, and aggregation attacks. We also discuss the security building blocks in our framework for integrating with FL.

\subsection{Federated Learning}
Federated Learning (FL) is a distributed machine learning paradigm designed to collaboratively train a shared global model across multiple clients (\textit{e.g.}, big institutions such as healthcare or financial organizations, mobile devices or sensors) without transferring their local datasets to a central server~\cite{mcmahan2016federated}. 
This framework is particularly useful as a first line of defense in privacy-sensitive applications, as it keeps raw data on-premises and only allows for the movement of the locally trained model.

Formally, consider $n$ clients, each holding a private dataset $D_i$. The goal is to train a global model $u \in \mathbb{R}^d$ that minimizes a global loss function expressed as a weighted average of local empirical losses:

\[
\min_{u} F(u) = \sum_{i=1}^n w_i F_i(u),
\]
where $F_i(u) = \frac{1}{|D_i|} \sum_{(x, y) \in D_i} \ell(u; x, y)$ is the local loss function of client $i$ based on its dataset $D_i$, and $\ell$ is a suitable loss function (\textit{e.g.}, cross-entropy or mean squared error). The weight $w_i = \frac{|D_i|}{\sum_{j=1}^n |D_j|}$ represents the proportion of data owned by the client $i$.

The Federated Averaging (FedAvg) algorithm~\cite{mcmahan2017communication} is the most widely used approach in FL. In each communication round: (1) the server broadcasts the current global model $u^t$ to a subset of clients; (2) each selected client performs local training, typically several steps of gradient descent: \[u_{i}^{t+1} = u^t - \eta \nabla F_i(u^t),\] where $\eta$ is the learning rate; (3) clients send their locally updated models $u_{i}^{t+1}$ to the server; (4) the server aggregates them to update the global model: \[u^{t+1} = \sum_{i=1}^n w_i u_{i}^{t+1}.\] This process iterates until convergence.

\subsubsection{Security Threats in FL}
Despite its privacy-conscious design, FL is vulnerable to several security threats that can compromise data confidentiality, model integrity, and overall utility. These threats are broadly categorized into poisoning, inference, and aggregation attacks.

%\smallskip
\paragraph{Poisoning Attacks} These attacks aim to corrupt the training process by injecting malicious behavior through either data or model updates. In data poisoning, the attacker manipulates their local dataset $D_i$ (e.g., by mislabeling or crafting samples) to influence the model's behavior in targeted or untargeted ways. In model poisoning, rather than altering data, the attacker directly crafts the local update $u^{(i)}$ to introduce backdoors or cause a targeted degradation of the global model. Some of the popular state-of-the-art poisoning attacks are label poisoning attack~\cite{10.5555/3327345.3327509}, backdoor attack~\cite{bagdasaryan2020backdoor}, LIE attack~\cite{baruch2019little}, Trim Attack~\cite{fang2020local} and Shejwalkar attack~\cite{shejwalkar2021manipulating}.
%\smallskip

\paragraph{Inference Attacks} These attacks exploit the visibility of intermediate or final model updates to extract private information. Membership Inference Attacks~\cite{nasr2019comprehensive} determine whether a specific data point was included in a client's training data. Property Inference Attacks~\cite{ganju2018property} reveal sensitive attributes or patterns from a client’s data (\textit{e.g.}, presence of a disease). Gradient Inversion Attacks~\cite{fredrikson2015model} reconstruct raw training data from shared gradients, especially in settings with few clients or sparse updates.

%\smallskip
\paragraph{Aggregation Attacks} 
These attacks target the aggregation mechanism itself, leading to unreliable or biased global models~\cite{guo2020v,wang2024breaking}. In traditional federated learning, the server is assumed to be honest but curious; however, this assumption may not always hold in practice. If the central server is malicious or compromised, it can deliberately perform an incorrect aggregation of client updates, thereby compromising the integrity and utility of the global model. A dishonest server may alter the aggregation logic, such as assigning skewed weights to specific clients, selectively dropping updates from certain clients (e.g., honest ones), or altering them before aggregation, which can lead to bias, reduced performance, or intentional misbehavior of the global model. A compromised server may collude with a subset of adversarial clients to amplify their impact by preferentially accepting their manipulated updates.

%\smallskip
\medskip
These vulnerabilities necessitate trustworthy Byzantine resilient aggregation in privacy-preserving settings. \textit{Ensuring privacy, integrity and trust in FL is crucial for its deployment in real-world applications.}

\subsubsection{Byzantine-Robust Aggregation Schemes in FL}\label{subsec-byz-rob} 
Byzantine-robust aggregation schemes are designed to protect the global model from malicious client updates. These schemes apply statistical criteria to filter, select, or reweight updates before aggregation. Below, we describe three representative schemes that differ in their approach: Krum selects a single most-consistent update, Trimmed Mean discards coordinate-wise extremes, and FLTrust weights updates by their alignment with a trusted reference. We later instantiate each of these within \sysname (Section~\ref{sec-proposed}).
%\smallskip

\paragraph{Krum} Krum~\cite{blanchard2017machine} is a Byzantine-robust aggregation rule that selects the client update most consistent with the majority of updates. This method is effective against up to $f < \frac{n}{2}$ Byzantine clients among $n$ total clients.
\smallskip

\noindent\textbf{Input:} A set of $n$ client updates $\{u_1, u_2, \dots, u_n\} \subset \mathbb{R}^d$ and a known upper bound $f$ on the number of Byzantine clients.

\noindent\textbf{Operations:} Krum performs the following operations.
\begin{enumerate}%[label=\textbf{Step \arabic*.}, leftmargin=*]
    \item \textbf{Pairwise distances:} For each client $i \in \{1, \dots, n\}$, compute the squared Euclidean distance to all other updates:
    \[
    d_{ij} = \|u_i - u_j\|^2, \quad \forall j \ne i.
    \]

    \item \textbf{Score computation:} For each $i$, sort the distances $d_{ij}$ in ascending order and take the $n - f - 2$ least values. 
    The \emph{Krum score} of $u_i$ is the sum of these distances:
    \[
    \mathrm{score}(u_i) = \sum_{j \in \mathcal{N}_i} d_{ij}, \quad \text{where } \mathcal{N}_i \subset \{1,\dots,n\} \setminus \{i\}, \,
    \]
    where $|\mathcal{N}_i| = n - f - 2.$
\end{enumerate}

\noindent
\textbf{Selection:} Return $u^{t+1}$, the update with smallest score.
\smallskip

\paragraph{Trimmed Mean} Trimmed Mean~\cite{yin2018byzantine} is a robust coordinate aggregation rule designed to tolerate up to $f$ Byzantine clients. It operates by discarding the extreme values in each coordinate across all client updates.

\smallskip
\noindent\textbf{Input:} A set of $n$ client updates $\{u_1, u_2, \dots, u_n\} \subset \mathbb{R}^d$ and a known upper bound $f$ on the number of Byzantine clients.

\noindent\textbf{Operations:} Trimmed-Mean does the following operations.
\begin{enumerate}
    \item \textbf{Coordinate-wise sorting:} For each coordinate $j \in \{1, \dots, d\}$, collect the $j$-th coordinate from all updates:
    \[
    V_j = \{ u_1[j], u_2[j], \dots, u_n[j] \}.
    \]
    Sort $V_j$ in ascending order.

    \item \textbf{Trim extremes:} Remove the largest $f$ and smallest $f$ values from $V_j$, for $f < n/2$.    
\end{enumerate}

\noindent\textbf{Selection:} Compute the mean of the remaining values:\[
    u^{t+1}[j] = \frac{1}{n - 2f} \sum_{v \in V_j^{\text{trimmed}}} v.
    \]
This is repeated independently for each coordinate $j$, yielding the aggregated update. %Trimmed-Mean is simple and effective.
\smallskip

\paragraph{FLTrust} FLTrust~\cite{cao2020fltrust} is a Byzantine-robust aggregation framework that assumes servers have access to a small, trusted dataset. It uses this dataset to both verify client updates and guide aggregation through a trust-weighted averaging scheme.
\smallskip

\noindent\textbf{Input:}
\begin{itemize}
    \item A set of client updates $\{u_1, u_2, \dots, u_n\} \subset \mathbb{R}^d$,
    \item A trusted update $u_0$ computed locally by the server on a small clean dataset.
\end{itemize}

\noindent\textbf{Operations:} FLTrust does the following operations.
\begin{enumerate}
    %\item \textbf{Normalize updates:} Normalize each client update and the server update to unit norm:
    %\[ \hat{u}_i = \frac{u_i}{\|u_i\|}, \quad \hat{u}_0 = \frac{u_0}{\|u_0\|}.\]

    \item \textbf{Compute trust score:} For each client $i$, compute trust score as the ReLU of cosine similarity between each local model update with the trusted update:
    \[
    w_i = \max\left\{0, \cos(u_i, u_0) \right\}.
    \]

    \item \textbf{Rescale updates:} Rescale each normalized client update by its norm and trust score:
    \[
    u_i^{\text{scaled}} = w_i \cdot \frac{\|u_0\|}{\|u_i\|} \cdot u_i.
    \]
\end{enumerate}

\noindent\textbf{Aggregation:} Aggregate the rescaled updates using a normalized trust-weighted average:
    \[
    u^{t+1} = \frac{\sum_{i=1}^{n} u_i^{\text{scaled}}}{\sum_{i=1}^{n} w_i}.
    \]

%FLTrust effectively filters out untrustworthy updates and anchors training to the trusted root update $u_0$.
\medskip

%Similar to previous algorithms, servers perform encrypted operations over private input from clients and public input of their own. Then, after obfuscation and getting plaintext shares, they compute the trust scores and rescale updates. Finally, the global model for the next iteration is selected using weighted averaging. %We provide a detailed example for this in Appendix~\ref{}.

\subsubsection{Related Work on FL Security}
%\subsection{Notes (to be deleted)}
%Refer these papers to design tables:
%\begin{enumerate}
%    \item SeaFlame, MudGuard, SafeFL, TAPFed, DefendFL
%    \item Robust \& Secure FL against hybrid attacks (TIFS)
%    \item Privacy Preserving \& Secure and robust FL (Wiley)
%    \item Robust \& Actively Secure Generative Collaborative learning (NIPS)
%    \item Trusted model aggregation with ZKP in FL (TPDS)
%    \item Byzantine Resilient Secure FL on Low Bandwidth NW (IEEE Access)
%    \item Decoding FL defenses: Systematization, Pitfalls \& Remedies (Arxiv)
%\end{enumerate}

%\begin{enumerate}
%    \item ACORN
%    \item A Feasible \& Scalable Malicious Secure Aggregation protocol for FL (TIFS)
%\end{enumerate}

%\begin{itemize}
%    \item SeaFlame: It uses secure aggregation protocol against malicious participants. It claims it is a communication-efficient protocol. It uses two non-colluding servers. 
%\end{itemize}

%FL has emerged as a foundational framework for privacy-preserving machine learning by enabling model training across distributed clients without the need to exchange raw data. Despite its advantages, FL protocols still have significant challenges: protecting the privacy of client data from leakage through local model updates, maintaining robustness of the global model in the presence of Byzantine (i.e., malicious or faulty) clients, and ensuring the integrity of computations performed by the central server. In this section, we provide a review of recent schemes that have been proposed to address these challenges.

Secure and Robust Federated Learning (SRFL) under a dual-server architecture aims to achieve privacy and robustness by separating aggregation and validation roles across two non colluding servers~\cite{liu2025secure}. SRFL employs CKKS encryption, random perturbation, and encrypted similarity-based verification models to filter poisoned updates. However, it is not extensible to other Byzantine-robust aggregation algorithms, nor does it provide cryptographic guarantees for the correctness of server-side aggregation. 
Using $\ell_2$ norm-bounding, ELSA proposes a safe aggregation method based on a dual-server model to guarantee privacy and strong gradient filtering to achieve Byzantine robustness \cite{gehlhar2023safefl}. SeaFlame builds upon ELSA by reducing communication costs while maintaining privacy and robustness against gradient boosting \cite{tang2025seaflame}. 
TAPFed employs threshold cryptography for resilience to dropouts and also serves as built-in verification, since it can tolerate some aggregators being malicious; as long as the number of honest parties exceeds a given threshold, the correct global model can be reconstructed by the clients \cite{xu2024tapfed}. However, it does not provide a Byzantine-robust aggregation for client updates.  RoFL \cite{burkhalter2021rofl}, Prio~\cite{corrigan2017prio}, and Prio+~\cite{addanki2022prio+} ensure Byzantine robustness through lightweight aggregation proofs but are not built for security against adversarial servers. SMPC-based secure aggregation is incorporated into SafeFL to ensure secure aggregation, but instead of providing server computations, it provides a guarantee that a single malicious server cannot corrupt the computation if there is at least one honest server \cite{gehlhar2023safefl}. 

MUDGUARD \cite{wang2024mudguard} and Franzese et al. \cite{franzese2023robust} represent state-of-the-art frameworks to holistically address the key challenges mentioned above. MUDGUARD guarantees these properties for the malicious minority on servers and for the malicious majority on clients. However, it lacks support for verifying client-side robustness enforcement on the server. 
\sysname improves on MUDGUARD by introducing a verifiability mechanism that enables clients to verify the correctness of the server's Byzantine-robust enforcement of the aggregated model. Franzese et al. proposed a fully decentralised framework that is based on SMPC, which is a flexible and theoretically sound approach, but it has limited practical deployment for large-scale FL tasks.

The proposed protocol \sysname addresses these shortcomings by introducing a privacy-preserving, verifiable, and secure aggregation framework. The system operates under a distributed-trust model with multiple servers, assuming that only one must be honest, and is designed to scale efficiently as the number of servers increases. Table~\ref{tab:study_comparison} presents a comparative analysis of key FL frameworks across privacy, robustness, verifiability, and threat model. While most prior work satisfies one or two of these properties, only \sysname offers comprehensive coverage. It is the only framework that guarantees client update privacy, Byzantine-robust aggregation, and full server-side verifiability with a minimal trust assumption and support for modern robustness measures.

\subsection{Multi-Key Fully Homomorphic Encryption}\label{sec:mkfhe}
Standard single-key Fully Homomorphic Encryption (FHE) suffices when all data is encrypted under a single public key by a single client. However, this is not sufficient for federated aggregation, which involves \emph{independent} clients.
Multi-Key FHE (MK-FHE) addresses this point by deriving a joint public key and joint relinearisation key from individual key shares, enabling homomorphic evaluation across ciphertexts encrypted by multiple parties while ensuring that no single party ever learns the collective secret~\cite{ma2022privacy,du2023efficient}. The construction below follows the Ring-LWE instantiation used for our protocol.

Let \(N\) be a power-of-two cyclotomic dimension with
\(q\in\mathbb Z\) the ciphertext modulus, and \(R_q=\mathbb Z_q[x]/(x^{N}+1)\).
Assume \(n\) parties (aggregation servers in our protocol) are 
\(\mathcal S_1,\dots,\mathcal S_n\).

\noindent(1) Individual secrets. Each party draws a small secret polynomial
\(s_k\overset{\$}{\leftarrow}\chi\subset R_q\)
and publishes a common random \(a\in R_q\).

\noindent(2) Collective public key. Each party samples an error \(e_k\leftarrow\chi\) and broadcasts
\(b_k:=-a s_k+e_k\).
The \emph{group public key} is
\[
\textsf{pk}= \bigl(b,\;a\bigr), \quad
b:=\sum_{k=1}^{n} b_k
      = -a\!\underbrace{\textstyle\sum_{k=1}^{n}s_k}_{S}+\;e,\;
      \;e=\sum_{k}e_k ,
\]
where the implicit \emph{group secret key}
\(S=\sum_{k}s_k\) is \emph{never} reconstructed.

\noindent(3) Collective relinearisation key.
\begin{itemize}
  \item Each server \(\mathcal S_k\) forms two plaintexts
        \(\alpha_k = s_k^{2}\) and \(\beta_k = 2s_k\),
        then encrypts them under the group public key:
        \[
          \textsf{ct}_{\alpha,k}= \Encop_{\textsf{pk}}(\alpha_k),\qquad
          \textsf{ct}_{\beta,k}= \Encop_{\textsf{pk}}(\beta_k).
        \]
  \item The ciphertexts are broadcast and homomorphically added to compute:
        \[
          \textsf{rlk}
              = \Encop_{\textsf{pk}}\!\bigl(S^{2}\bigr)
        \]
        where \(S=\sum_{k=1}^{n}s_k\) is the (implicit) group secret.
\end{itemize}

The value \(\textsf{rlk}\) is the \emph{group relinearisation key}, enabling efficient
ciphertext–ciphertext multiplication while revealing no individual secret key material.

\noindent(4) Encryption.
Anyone can encrypt \(m\in R_q\) via
\[
\Encop_{\textsf{pk}}(m):
\;u\leftarrow\chi,\;
e_1,e_2\leftarrow\chi,
\]
\[
(c_0,c_1)=(b u+e_1+m,\;a u+e_2).
\]

\noindent(5) Homomorphic operations. Addition is component-wise. Multiplication uses \textsf{rlk}:
\((c_0,c_1)\!\ast\!(d_0,d_1)\xrightarrow{\textsf{rlk}}
(\tilde c_0,\tilde c_1)\).

\noindent(6) Threshold decryption.
Each server outputs a \emph{partial decryption}
\(d_k=s_k c_1\).
Aggregating yields
\[
m = c_0 + \sum_{k=1}^{n} d_k \pmod q,
\]
so \emph{all} servers must cooperate; a single honest party
can block decryption and thus prevents data leakage to colluding
adversaries.

\smallskip
\noindent\textbf{Security Building Blocks.}
The confidentiality and correctness of homomorphic evaluation rely on the \emph{decisional Ring–LWE} assumption, the standard hardness foundation for lattice-based fully homomorphic encryption (BGV/CKKS).  

\begin{comment}
    \subsubsection{Hardness Assumptions and Encryption Schemes}

Our security claims rest on two well‑studied computational assumptions: the \emph{Decisional Ring‑Learning‑With‑Errors (Ring‑LWE)} problem underlying lattice‑based FHE, and the classical \emph{Discrete Logarithm Problem (DLP)} underpinning the per‑coordinate commitments.
\end{comment}

\subsubsection{The Decisional Ring‑LWE Problem.}
Let $N$ be a power of two; $q$ a prime with $q\equiv1\pmod{2N}$; and $R_q=\mathbb{Z}_q[x]/(x^N+1)$ the degree‑$N$ cyclotomic ring. Fix a noise distribution $\chi$ over $R_q$ (\textit{e.g.}, a centered binomial or discrete Gaussian) with small standard deviation. For a secret $s\overset{\$}{\leftarrow}\chi$ and independently sampled $a\overset{\$}{\leftarrow}R_q$, $e\overset{\$}{\leftarrow}\chi$, an \emph{RLWE sample} is the pair
\[
(a,\;b = a\cdot s + e) \;\in\; R_q\times R_q.
\]
The \emph{decisional RLWE} problem asks to distinguish polynomially many independent RLWE samples from uniformly random pairs in $R_q\times R_q$. Formally, the advantage of a probabilistic polynomial‑time distinguisher $\mathcal{D}$ is,
\begin{align*}
\mathrm{Adv}_{\mathcal{D}}^{\mathrm{RLWE}} = {} &
\left| \Pr[\mathcal{D}(a,b)=1 \mid b = a \cdot s + e] \right. \\
& \qquad \left. - \Pr[\mathcal{D}(a,u)=1 \mid u \overset{\$}{\leftarrow} R_q] \right|
\end{align*}

The \emph{Decisional RLWE Assumption} states that for any PPT $\mathcal{D}$ this advantage is negligible in the security parameter.

\subsection{Verifiability of Arithmetic Computations} \label{subsec:verify}

\subsubsection{Commitment-based verification}
To enable public verification of basic arithmetic on committed values, each client publishes discrete logarithmic commitments to their data. Suppose a client sends plaintext values \(x, y \in \mathbb{Z}_p\) to the server, along with corresponding commitments: \[h_1 = g^x, \qquad h_2 = g^y,\] where \(g\) is a generator of a cyclic group \(\mathbb{G}\) of prime order \(p\). It satisfies the following properties:

\smallskip\noindent\textbf{Addition (\(a = x + y\)).}\;
  The server computes \(a = x + y\). Any verifier can confirm the result using the group law:
  \[
  g^a \stackrel{?}{=} h_1 \cdot h_2.
  \]

\smallskip\noindent
  \textbf{Multiplication (\(b = x \cdot y\)).}\;
  The server computes \(b = x \cdot y\). Verification is done via a bilinear pairing \(e : \mathbb{G} \times \mathbb{G} \to \mathbb{G}_T\):
  \[
  e(g, g)^b \stackrel{?}{=} e(h_1, h_2).
  \]

%These checks ensure arithmetic integrity without revealing the underlying plaintexts, assuming the hardness of the Discrete Logarithm Problem (DLP) and the security of the pairing group.

\smallskip
\noindent\textbf{Security Building Blocks.} Soundness of the server-side verification mechanism is guaranteed by the intractability of the \emph{Discrete Logarithm Problem} in a prime-order group, which underlies the exponentiation commitments attached to every model coordinate.
%By pairing RLWE-based encryption with DLP-based verifiability, the protocol simultaneously achieves post-quantum privacy and publicly verifiable integrity properties formalized in the subsections that follow.

\subsubsection{Discrete Logarithm Problem (DLP)}
For the commitments, $h_j=g^{u_{i,j}}$ works in a cyclic prime order group $(\mathbb{G},\cdot)$ with the generator $g$. Given $(g,\,g^x)$, an adversary’s task is to recover $x$. The \emph{Discrete Logarithm Problem} asserts that this inversion is infeasible for any PPT algorithm when the group order is sufficiently large (\textit{e.g.}, a 256-bit prime). The integrity of verification, therefore, reduces to DLP hardness.

In summary, the confidentiality of the ciphertexts is based on decisional RingLWE, while the extractibility of the commitments and the integrity of the server computations are based on DLP. These assumptions are standard and widely deployed, providing strong confidence in the cryptographic foundations of our framework.

\section{Problem Overview}
%In the case of only two clients, a single malicious client colluding with the malicious server can subtract the colluding client’s known update from the aggregated result or inject poisoned updates to degrade the global model.

We consider a federated learning (FL) framework comprising $n$ clients and $m$ servers, with $n \gg m \geq 2$. %Each client $i \in \{1, 2, \ldots, n\}$ holds a private dataset $\mathcal{D}_i$ and seeks to collaboratively train a global model without sharing its raw data. The learning process is orchestrated through multiple rounds, during which clients compute local model updates based on their data and send them to the servers in encrypted form.
%Their roles are defined as below.

\smallskip\noindent
\textbf{Client Role.} Each client performs local training on its dataset and generates a model update vector $u_i$. To preserve privacy, the client encrypts its update using a multi-key fully homomorphic encryption (MK-FHE) scheme (with decryption keys shared among servers), yielding a ciphertext $\Enc{u_i}$ that supports computations over it. %The client then splits the ciphertext into $m$ disjoint shares, one for each server, denoted as $\{\Enc{u_{i1}}, \Enc{u_{i2}}, \ldots, \Enc{u_{im}}\}$, and sends them to the respective servers.

\smallskip\noindent
\textbf{Server Role.} Each server $k \in \{1, 2, \ldots, m\}$ is responsible for performing secure aggregation over the encrypted shares. By leveraging homomorphic properties of MK-FHE, servers can jointly compute statistical functions (e.g., sum, average, norm) over ciphertexts without decrypting them. %The servers engage in interactive protocols to perform operations such as noise obfuscation and masking via encrypted random numbers, thereby randomizing the intermediate ciphertexts. Once the necessary encrypted computations are completed, the servers engage in threshold decryption to reveal masked plaintext results, which can then be used for further plaintext operations.

\begin{comment}
We define the problem setting consistent with standard assumptions in Federated Learning. The protocol involves two main roles:

\textbf{Participants (Clients):}
\begin{itemize}
  \item Each participant trains a local model on private data and generates an encrypted update.
  \item Participants transmit a subset of their encrypted updates and accompanying commitments to designated servers.
\end{itemize}

\textbf{Servers:}
\begin{itemize}
  \item Each server collects encrypted model updates from the participating clients.
  \item Servers perform aggregation operations over the received updates.
  %\item After aggregation, a server produces a cryptographic proof (via exponentiation commitment checks) to convince the other server of the correct computation.
  %\item Only after successful verification, the servers engage in collaborative multi-key FHE decryption.
\end{itemize}
\end{comment}

\subsection{Threat Assumptions} We consider a malicious threat model in which both servers and clients may adversarially deviate from the protocol to compromise privacy or correctness. %It is subject to the following consideration:
The system consists of multiple clients (\textit{e.g.}, healthcare and financial institutions) and two (or more) coordinating servers for managing the FL process. Each client maintains a private dataset and periodically submits local model updates. \textbf{Any subset of clients may also be malicious}, which can submit poisoned or arbitrarily crafted local updates to bias or degrade the global model or attempt to extract information about an honest client’s updates by colluding with the malicious server or by analysing the change in successive global models.

The protocol assumes that \textbf{at least one of the servers is honest}, while other servers may even collude and (1) attempt to learn individual client updates by deviating from the protocol or observing data flows, (2) manipulate the global aggregation process by selectively dropping or altering updates, or (3) send incorrect global models back to the clients, potentially disrupting model convergence. 

If a malicious server ignores all updates or selectively drops certain client updates, the protocol might fail to produce a global model. The protocol must detect such behaviour and prevent the malicious server from gaining information by observing which updates are included.
%\smallskip

%\begin{itemize}
 % \item We assume that at least one server is \emph{honest-but-curious}. This server follows the protocol faithfully but attempts to learn as much information as possible from its inputs.
 % \item Other servers may be \emph{malicious}, deviating arbitrarily from the protocol, including falsifying aggregation results or forging commitments.
 % \item The protocol guarantees that any deviation by a malicious server will be detected by the honest one through the verifiability mechanism.
%\end{itemize}

\subsubsection{Security Guarantees} Our proposed framework offers the following security guarantees:
\begin{enumerate}
    \item \textbf{Privacy of client updates:} Clients' local model updates remain protected throughout the training process. It ensures that even if a subset of clients or servers is compromised, the confidentiality of client data is preserved.
    
    \item \textbf{Integrity of client updates:} Ensuring that the system can detect and tolerate poisoned updates received from the malicious clients.
    
    \item \textbf{Integrity of Aggregation:} The aggregation process is verifiable and resistant to tampering, ensuring that the final model reflects genuine client contributions.
    
    \item \textbf{Scalability and Efficiency:} By distributing computational tasks and transitioning certain operations to the plaintext domain when secure to do so, the framework achieves practical efficiency suitable for large-scale deployments.
\end{enumerate}

\begin{comment}
\begin{itemize}
  \item \textbf{Confidentiality.} All linear operations occur on ciphertexts; masking by \(r\) hides magnitude information until the final division, which occurs post‑decryption.
  \item \textbf{Integrity \& verifiability.} Exponentiation commitments ensure every coordinate is well‑formed; breaking verification requires solving the discrete logarithm problem.
  \item \textbf{Robustness.} Cosine‑similarity weighting matches the original FL‑Trust, limiting the influence of poisoned updates.
  \item \textbf{Collaborative Decryption Security.} Since the system employs multi‑key FHE, decryption requires the cooperation of all servers. An honest server will refuse to complete the decryption phase unless it successfully verifies the computation. \textit{This ensures that malicious servers cannot unilaterally decrypt the model or infer any meaningful information from the ciphertexts}.
\end{itemize}

Collectively, the protocol achieves privacy, correctness, and robustness guarantees while remaining compatible with practical FHE frameworks that restrict operations to element‑wise addition and multiplication.
Under this model, correctness and integrity are ensured so long as one server remains honest. 
\end{comment}

\section{Our proposed framework: \protect\sysname}\label{sec-proposed}

This section presents our \textbf{P}rivate, \textbf{Ro}bust, and \textbf{Ve}rifiable \textbf{FL} (\sysname) framework. 
%Our framework facilitates the execution of Byzantine-robust aggregation schemes to ensure data integrity while maintaining confidentiality. We propose a multiserver architecture to address the scalability limitations of homomorphic encryption-based methods and to guarantee the integrity of the aggregated result, even if only one server is honest. 
Since each Byzantine-robust aggregation scheme makes different assumptions about malicious settings and employs unique approaches, we demonstrate \sysname using state-of-the-art poisoning-resistant algorithms such as Krum, Trimmed-Mean, FLTrust, and discuss its extendability to MESAS, CrowdGuard, and FLAME. We intentionally selected these algorithms because they cover various types of operations, including statistical operations, client-provided private input, and server-provided public input.
%But, rather than hard-binding our framework to these algorithms, we propose a modular design to cover a broad class of Byzantine robust schemes in a privacy-preserving and verifiable way. 

%in Figure~\ref{fig:modular-fhe-framework} (as described below).  
%\smallskip

\begin{figure}
    \centering
    \includegraphics[width=0.5\textwidth]{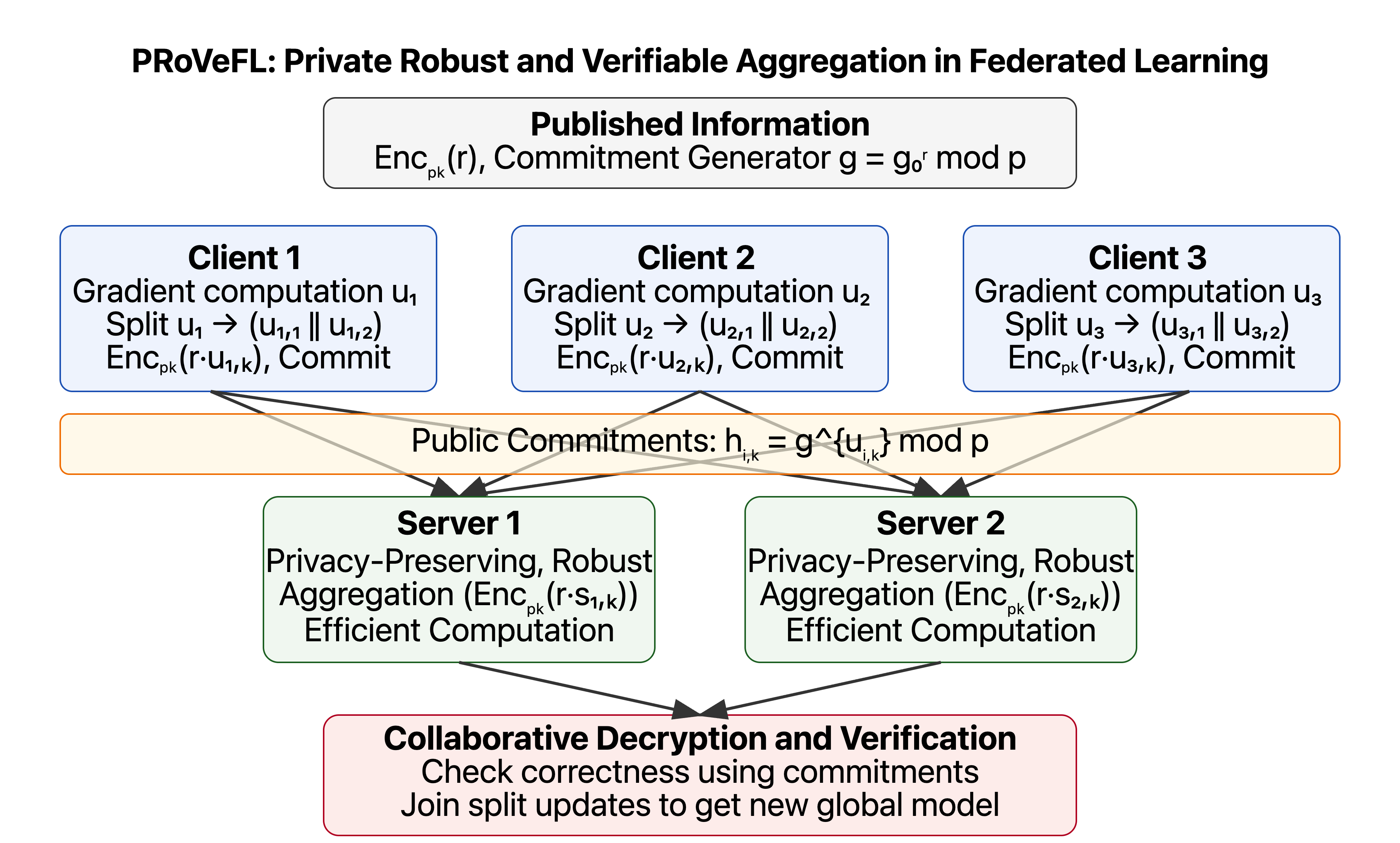}
    \caption{Our Proposed Framework}
    \label{fig:provefl_diag}
\end{figure}

\begin{comment}
The framework meets three simultaneous goals:

\begin{enumerate}
  \item \emph{privacy} - ensures that client updates remain encrypted at all times and also protects intermediate computations to not reveal any related sensitive information;
  \item \emph{robustness}—aggregation tolerates up to $f$ Byzantine
        clients under well-studied defences such as
        Krum~\cite{blanchard2017machine},
        Trimmed-Mean~\cite{yin2018byzantine}, and
        FLTrust~\cite{cao2020fltrust}; and
  \item \emph{verifiability}—any party can validate every
        server-side arithmetic operation by means of bilinear
        pairings and public commitments.
\end{enumerate}
\end{comment}

%\sysname is a modular, multi-server architecture that supports arbitrary Byzantine-robust aggregation rules under end-to-end multi-key homomorphic encryption (MK-FHE). To lift the scalability bottleneck of single-server FHE, \sysname distributes ciphertexts across two or more servers.
%Confidentiality holds as long as \emph{one} server is honest;integrity is enforced by aborting whenever a pairing check fails.Figure~\ref{fig:modular-fhe-framework} illustrates the protocol flow and shows how any poisoning-resistant scheme can be plugged into the same encrypted pipeline.

%\smallskip
%\noindent\textbf{Notation.}
\subsection{System Model and Notation}

%\subsubsection{Notation}
We use the following notation throughout this section. Let \( n \) denote the number of clients and \( f \) the upper bound on Byzantine (malicious) clients. Each client \( i \in \{1, \dots, n\} \) holds a local model update vector \( u_i \in \mathbb{R}^N \), computed as the gradient \( \nabla\mathcal{L}_i \) of a local loss function. Let $m$ denote the number of servers. The vector is deterministically split into $m$ equal-length parts: \( u_i = (u_{i,1} \,\|\, u_{i,2}\,\|\,\dots \,\|\,u_{i,m}) \), with \( u_{i,k} \in \mathbb{R}^{N/m} \) for \( k \in \{1,2, \dots , m\} \).  $\Encop_{\mathsf{pk}}(x)$ denotes the encryption of $x$ with the public key $pk$  where the public key is generated as described in Section~\ref{sec:mkfhe}.

Server-held ciphertexts are denoted \( c_{i,k} \), representing the \( k^\text{th} \) share of the encrypted update vector for client \( i \) held with Server-$k$. For schemes such as FLTrust, the trusted server model update is denoted \( u_0 \) with norm \( \|u_0\| \), and the trust score for client \( i \) is denoted \( w_i \). Final aggregation steps produce the next-round global model update \( u^{t+1} \). All cryptographic operations assume cyclic group \( \mathbb{G} \) of prime order \( p \), with associated bilinear map \( e: \mathbb{G} \times \mathbb{G} \rightarrow \mathbb{G}_T \).

Figure~\ref{fig:provefl_diag} demonstrates \sysname considering the example of two servers and three clients. 
Details of the framework are described below.

%\smallskip
%\noindent\textbf{Overview.} 
%\smallskip
\subsection{Protocol Workflow} 
\sysname proceeds in three stages: generation of published parameters, client-side and server-side operations. 
%The remainder of this section proceeds in three layers, each illustrated by a dedicated figure.
For ease of readability, we consider only two servers to describe the protocol, but it can be extended to any fixed number of servers. 
We first explain how the two servers jointly derive the global random mask and commitment generator; this protocol is summarized in Figure~\ref{fig:pub-params}.
We then present the \emph{generic} encrypted aggregation pipeline, covering client uploads, server-side homomorphic processing, collaborative decryption, public verification, and aggregation, as shown in Figure~\ref{fig:modular-fhe-framework}.
Finally, we instantiate the ``method-specific'' block of that pipeline with three state-of-the-art defenses: Krum (Figure~\ref{fig:modular-fhe-krum}), Trimmed-Mean (Figure~\ref{fig:modular-fhe-tm}), and FLTrust (Figure~\ref{fig:modular-fhe-fltrust}).
It demonstrates that a single set of cryptographic building blocks underlies all three Byzantine-robust aggregation schemes while accommodating their distinct statistical logics. 
%Now, we discuss each stage in \sysname in detail.

\begin{figure}[!htbp]
\setlength{\fboxsep}{.9pt}
\begin{center}
\begin{tcolorbox}[enhanced,
    drop fuzzy shadow southwest, colframe=black,colback=white]
\footnotesize
\noindent\textbf{Generation of Published Parameters}

\noindent\textbf{Setup:}
Fix a public generator \(g_0\in\mathbb G\) of prime order \(p\) and the joint FHE public key \(\mathsf{pk}\).

\vspace{4pt}
\noindent\textbf{Step 1: Local randomness.}
\[
  r_k \overset{\$}{\leftarrow} \mathbb Z_p
  \quad\text{by each server } \mathcal K.
\]

\noindent\textbf{Step 2: Share publication.}
\[
  c_k = \Encop_{\mathsf{pk}}(r_k),
  \qquad
  g_k = g_0^{r_k} \bmod p.
\]

\noindent\textbf{Step 3: Homomorphic aggregation.}
\[
  \Encop_{\mathsf{pk}}(r)
  = c_1 + c_2,
  \quad
  r = r_1 + r_2 \bmod p.
\]

\noindent\textbf{Step 4: Commitment generator.}
\[
  g = g_1 \cdot g_2 = g_0^{r_1+r_2} = g_0^r \pmod{p}.
\]

\vspace{4pt}
\noindent\textbf{Published information:}
\begin{itemize}
  \item Ciphertext of the global mask: \(\Encop_{\mathsf{pk}}(r)\).
  \item Commitment generator: \(g = g_0^r \bmod p\).
\end{itemize}

\end{tcolorbox}
\end{center}
\caption{Distributed Generation of Mask and Commitment Generator (described for 2 Servers)}
\label{fig:pub-params}
\end{figure}

%Replacement -1
\begin{comment}

\begin{figure}[!htbp]
\setlength{\fboxsep}{.9pt}
\begin{center}
\begin{tcolorbox}[enhanced,
    drop fuzzy shadow southwest, colframe=black,colback=white]

\noindent\textbf{Generation of Published Parameters \textit{[New]}}

\noindent\textbf{Setup:}
Fix a cyclic group \(\mathbb G\) of prime order \(p\), independent generators
\(g_0,h\in\mathbb G\) with unknown \(\log_g h\), and the joint FHE public key
\(\mathsf{pk}\). Let
\[
    \mathsf{Com}(x;\rho)=g_0^x h^\rho .
\]

\vspace{4pt}
\noindent\textbf{Step 1: Local randomness.}
Each server \(S_k\), for \(k\in\{1,2\}\), samples
\[
    r_k\overset{\$}{\leftarrow} \mathbb Z_p .
\]

\noindent\textbf{Step 2: Share publication.}
\[
\begin{aligned}
    c_k = \Encop_{\mathsf{pk}}(r_k),\quad g_k  =g_0^{r_k} \mod{ p}.
\end{aligned}
\]

\noindent\textbf{Step 3: Homomorphic aggregation.}
\[
\begin{aligned}
    \Encop_{\mathsf{pk}}(r)
        &= c_1+c_2,
        & r &= r_1+r_2 \bmod p.
\end{aligned}
\]
\noindent\textbf{Step 4: Commitment generator.}
\[
  g= g_1 \cdot g_2 = g_0^{r_1+r_2} = g_0^r \pmod{p}.
\]

\vspace{4pt}
\noindent\textbf{Published information:}
\begin{itemize}
  \item Ciphertext of the global mask: \(\Encop_{\mathsf{pk}}(r)\).
  \item Commitment generator: \(g= g_0^r \bmod p\).
\end{itemize}

\end{tcolorbox}
\end{center}
\caption{Distributed Generation of Mask and Commitment Parameters (described for 2 Servers)}
\label{fig:pub-params}
\end{figure}

\end{comment}
% End Replacement -1

%2nd Original Table
\begin{figure}[t!]
\setlength{\fboxsep}{.9pt}
\begin{center}
    \begin{tcolorbox}[enhanced,
    drop fuzzy shadow southwest, colframe=black,colback=white]
\footnotesize

\noindent\textbf{Published Information:} 

\begin{itemize}
    \item FHE Encryption of Random Mask $r$: $\Encop_{\mathsf{pk}}(r)$
    \item Commitment Generator: $g= g_0^{r} \mod p$
\end{itemize}

\noindent\textbf{Client Side Operations:} 
Each client performs the six deterministic steps
\begin{itemize}
    \item  \textit{Gradient Computation}: The client computes its local model update
  \( u_i \leftarrow \nabla\mathcal{L}_i \).
  \item \textit{Deterministic Split}: The update vector is split into equal halves (assuming 2 servers) $u_i = (u_{i,1} \,\|\, u_{i,2}) \in \mathbb{R}^{N}, \quad \text{with } u_{i,1}, u_{i,2} \in \mathbb{R}^{N/2},\; $
  \item \textit{FHE encryption of split-updates}: for $k = 1, 2$,
$c_{i,k} = u_{i,k}. \Encop_{\mathsf{pk}}(r)=\Encop_{\mathsf{pk}}(r \cdot u_{i,k})$.
  \item  FHE Encryption of norm (Only for FLTrust)
$e_{i,k} = \Enc{ \frac{r}{\|u_{i,k}\|}}$ [Note - It can also be calculated privately at the server side.]
  \item \textit{Commitment Generation}:
$ h_{i,k_l} = g^{u_{i,k_l}} \bmod p,~\text{for } k = 1, 2 \text{ and } l= 1,2, \dots N$.
  \item \textit{Transmission}: The client sends its encrypted split updates to the respective servers and publishes committed data to everyone. 
\end{itemize}

\noindent\rule{\textwidth}{0.1pt}

\noindent\textbf{Server Side Encrypted Operations:} 
\begin{enumerate}
    \item \textit{Method-Specific Encrypted Aggregation}: the servers invoke one of the sub-routines in any Byzantine-robust aggregation scheme.
    \item \textit{Collaborative Decryption}: The results are decrypted by collecting partial decryption keys from all servers.
    \item \textit{Verification of Correct Encrypted Aggregation}:
    \begin{enumerate}
        \item The method described in Section \ref{subsec:verify} is invoked by another server to verify the correctness of aggregation.
        \item The honest servers abort if the verification fails at this moment.
    \end{enumerate}
\end{enumerate}
\noindent\rule{\textwidth}{0.1pt}

\noindent\textbf{Server Side Plaintext Operations:}
\begin{enumerate}[resume]
    \item If verification succeeds, join the aggregated split updates and compute the final trust score.
    \item Selection: This step is computed in the same way as in plaintext, similar to the weighted averaging adopted in any Byzantine-robust algorithm.
\end{enumerate}
    
\end{tcolorbox}
\end{center}
\caption{\sysname Working Mechanism.} 
\label{fig:modular-fhe-framework}
\end{figure}

\subsubsection{Initialisation} To obscure structural patterns in encrypted client updates and defend against inference attacks, all servers collaboratively generate an encrypted random multiplier~$\Enc{r}$ (Figure~\ref{fig:pub-params}) by jointly contributing secret shares, ensuring that the plaintext $r$ remains unknown to any single party. Each client homomorphically applies $\Enc{r}$ to its local update~$\Enc{u_{ik}}$, yielding randomized ciphertexts $\Enc{r \cdot u_{ik}}$. 
%This randomized masking step statistically de-correlates encrypted inputs from the original updates, preventing reverse-engineering attempts.
This randomized masking step ensures that encrypted updates are obfuscated by an independently sampled multiplier, effectively breaking any consistent statistical correlation with the original plaintexts. As the random factor changes across rounds and remains secret, adversaries cannot accumulate sufficient information to reverse engineer individual updates or inter-client relations.

%Each server receives a distinct subset of encrypted model updates from all clients and independently performs statistical computations and method–specific aggregation on these ciphertexts, resulting in encrypted scores $\Enc{r \cdot s_{ik}}$. These randomized ciphertexts are then decrypted across servers using collaborative decryption, revealing only the obfuscated scores $r \cdot s_{ik} $, which cannot be reverse engineered to recover any client's model update or reveal inter-client statistics due to the presence of more unknowns than known. 

\subsubsection{Privacy-preserving aggregation}
Each server receives a distinct subset of these encrypted masked updates and independently performs the encrypted operations in any method-specific aggregation scheme, resulting in ciphertexts~$\Enc{r \cdot s_{ik}}$, where $s_{ik}$ is the plaintext value of the applied statistical operation. These are then jointly decrypted using threshold decryption, revealing only the obfuscated values $r \cdot s_{ik}$. 
This relies on the fact that $r$ is 
unknown and freshly sampled each round, as discussed above. 
%these intermediate values do not leak client-specific information and cannot be inverted. 

\subsubsection{Verification} 
Verification follows the commitment-based framework described in Section~\ref{subsec:verify}.
Each aggregation primitive produces a publicly verifiable proof. 
%After decryption, each server uses a non-interactive proof of well-formedness to generate proofs for the aggregated results. 
The peer server verifies these proofs by checking, for every coordinate~$\ell$, that the pairing or group relation holds, e.g.,
\(
  g^{d_{i,\ell}} \stackrel{?}{=} h_{i,\ell}
\)
(or the corresponding additive check for a method) holds.  If \emph{any} check fails, then the honest server
refuses to release its decrypted share, causing the protocol
to abort and thus prevent malformed aggregation. %Hence a malicious server cannot force acceptance of an incorrect aggregate. 

\subsubsection{Collaborative decryption} Provided that all verifications succeed, the servers send the decrypted shares (of intermediate calculations), and concatenate to reconstitute each full score vector. The masked plaintext aggregates are then processed exactly as in their classical counterparts: Krum selects the vector with the minimal score, Trimmed-Mean discards extreme coordinates and averages the remainder, while FLTrust computes cosine-similarity weights $w_i$ and outputs the normalised weighted average. These final operations involve only inexpensive plaintext arithmetic (minimal ciphertext arithmetic) and therefore add negligible overhead compared with the encrypted phase.

%This careful transition between secure encrypted computation and strategic exposure of obfuscated plaintext values enables the protocol to maintain \textit{strong confidentiality guarantees, ensure the integrity of updates, and scale to a large number of clients without incurring prohibitive computational costs.}
%\smallskip

\subsubsection{Strengthened Security Model} The proposed framework operates under a strengthened security model enabled by MK-FHE under a distributed multi-server architecture, which collectively provides strong guarantees of confidentiality and correctness. A critical security concern in this setting is the possibility of collusion between malicious clients and a compromised server: since clients know their own updates $u_i$, a colluding client-server pair observing the masked value $r \cdot u_i$ could recover the global mask $r$ by simple division, subsequently unmasking all other honest clients' intermediate statistics. 
\sysname addresses this through an additive blinding mechanism inspired by the order-preserving encryption scheme of Dyer et al.~\cite{divisors2017order}, wherein each intermediate masked statistic is additively blinded by an independently sampled $\delta_i \xleftarrow{\$} \{0, \ldots, B-1\}$. The bound $B$ is chosen to satisfy two constraints simultaneously: $B < r \cdot \epsilon / 2$ to ensure ordering is preserved across distinct client updates, and $B > 2^\lambda$ to ensure brute-force recovery of $r$ is computationally infeasible. These constraints are jointly satisfiable by choosing $r > 2^{\lambda+1}/\epsilon$, which is achievable within standard cryptographic group sizes, and guarantee both the correctness of ordering-dependent aggregation rules and security against mask recovery by any coalition of colluding clients and servers. %Beyond this, the collaborative randomization and encrypted computation prevent any server, even colluding subsets, from inferring client updates, while the framework enables efficient integration of verification mechanisms without incurring the prohibitive cost of executing such checks over fully encrypted data. Security guarantees remain intact as long as at least one server is honest.

\subsection{Core Aggregation Primitives}

Rather than a collection of independent protocol instantiations, \sysname is a general framework for privacy-preserving Byzantine-robust aggregation. 
The framework supports a number of cryptographically-verifiable aggregation patterns that enable many existing schemes to be supported directly. 
Each primitive operates on a collection of masked, encrypted values, and produces outputs that can be verified before their use in subsequent aggregation steps. 
The core primitives are as follows: 

\subsubsection{Linear aggregation (LA)}
Many robust aggregation methods involve linear operations on client updates, i.e., addition, subtraction, (weighted) averaging, and norm computation. 
These can be handled in an encrypted space: given the encrypted client updates 
${\text{Enc}_{pk}(r \cdot u_i)}$, the servers can evaluate expressions of the form
$\sum_i \alpha_i u_i$
homomorphically, where the parameters $\alpha_i$ may be fixed constants, or public values derived elsewhere within the aggregation process. 
%This primitive underlies standard FedAvg aggregation, norm clipping, trust-weighted averaging in FLTrust, and the final aggregation stage of many robust defenses.

\subsubsection{Pairwise statistical evaluation (PSE)}
Many Byzantine-robust aggregation rules are based on examining all pairwise relationships between client updates, expressed as (low-degree) polynomial functions on pairs of values.
This captures (squared) Euclidean distance, inner products, cosine similarity and other coordinate-wise differences. 

For two encrypted updates $u_i$ and $u_j$, servers compute encrypted statistics of the form
\[
\text{Enc}_{pk}(r^c \cdot f(u_i,u_j)),
\]
\noindent
where $f(\cdot)$ is the low-degree polynomial function and $c$ captures the degree of the operation.
Higher-order statistics can also be handled (e.g., considering all triples of values), but are not popular in practice due to the higher costs (cubic or worse in the number of clients). 

\subsubsection{Secure comparison and ordering (SCO)}
A common pattern in robust aggregation is to consider values in sorted order or to find a value.

\sysname allows secure comparison via computing masked values, which can then be collaboratively decrypted to support order-preserving statistics. Given two values $x$ and $y$, servers can obtain only the masked difference $r(x-y)$,
which reveals the sign of $(x-y)$ while obscuring the underlying values. 
This also enables more complex patterns, such as coordinate-wise sorting via compare-and-swap.

\subsubsection{Selection and Filtering (SF)}
The ability to compare and sort also enables various trimming and pruning rules based on top-$k$ or bottom-$k$. After the required statistics have been computed and verified, \sysname performs selection in the plaintext domain using only masked intermediate values.
Since the random mask preserves the relative ordering, the selection and filtering decisions match those applied to the original plaintext values.

\subsection{Instantiation of Byzantine-Robust Aggregation Schemes} 
%Based on the above primitives, then 
Any aggregation scheme can be implemented within \sysname if it can be expressed in terms of a bounded number of applications of the above primitives (LA-linear aggregation, PSE-pairwise statistical evaluation, SCO-secure comparison and ordering, SF-selection and filtering). 
To exhibit its utility, 
%Having discussed the modular framework, 
we apply this principle to the design of three state-of-the-art Byzantine-robust algorithms in \sysname: Krum~\cite{blanchard2017machine}, Trimmed-Mean~\cite{yin2018byzantine}, FLTrust~\cite{cao2020fltrust}, and MESAS~\cite{krauss2023mesas} defined in Section~\ref{subsec-byz-rob}.

\begin{figure}[t!]
\setlength{\fboxsep}{.9pt}
\begin{center}
\begin{tcolorbox}[enhanced,
   drop fuzzy shadow southwest, colframe=black,colback=white]
\footnotesize

\noindent\textbf{Input:} Each server $k$ receives a set of $n$ encrypted client updates $\{c_{i,k} = \Enc{r\cdot u_{i,k}}: i=1,2,\dots,n\}$, where $u_{i,k}$ is $k^{th}$ subset of client update of client \textit{i}. %$f$ is the bound on the number of Byzantine clients.

\noindent\rule{\textwidth}{0.1pt}

\noindent\textbf{Server-Side Encrypted Operations}
\begin{enumerate}
\item For every ordered pair \((i,j)\) the servers compute
\[
\Enc{d_{ij_{k}}} = (c_{i,k}-c_{j,k})^2,  \forall j \ne i.
\]
Note that,
\[
\Enc{d_{ij_{k}}}=\Enc{(r \cdot x_{i,k}-r \cdot x_{j,k})^2}.
\]
\item \textit{Collaborative decryption.}
Decrypt to obtain
\[
d_{ij_k}=(r \cdot (x_{i,k}-x_{j,k}))^2.
\]

\item []\textbf{Verification of Correct Computation.}
\item Servers publish $e(g_0,g_0)^{d_{ij,k_l}}$ as a proof for the correct computation.
\item Peer Servers: For every coordinate \(\ell\) and pair \((i,j)\) verify,
\[
  e(h_{ik_{\ell}}(h_{jk_{\ell}})^{-1},h_{ik_{\ell}}(h_{jk_{\ell}})^{-1}) \stackrel{?}{=} e(g_0,g_0)^{d_{ij,k_l}}.
\]

\noindent\rule{0.9\textwidth}{0.1pt}

\noindent\textbf{Plaintext operations}
\item \textbf{Pairwise distances:} Now, any of the servers (assuming at least one is honest) compute the pairwise distance with all decrypted $d_{ij_{k}}$. The distances are scaled up by $r$. However, they preserve order, which is sufficient for Krum.
    \item \textbf{Score computation:} This step is computed as if with plaintext, as the sum of $n-f-2$ distances, for each client.
    \item \textbf{Selection:} This step can be computed in the same way as in plaintext, selecting the client update with the smallest score, as the next global model update.
\end{enumerate}

\end{tcolorbox}
\end{center}
\caption{Instantiation of Krum in \sysname} 
\label{fig:modular-fhe-krum}
\end{figure}

\begin{comment}
     There remains a concern that all distances, though scaled, would still follow the same distribution. To avoid any possible adaptive inferences, we suggest choosing different $r$ in each iteration.
\end{comment}

\subsubsection{\sysname-enabled KRUM (Figure~\ref{fig:modular-fhe-krum})}

% After receiving the local updates from clients, servers perform encrypted operations in steps 1-3, and then, after doing obfuscation, they decrypt to get the plaintext shares and perform the remaining operations to calculate the final scores in steps 4-6. Finally, the global model for the next iteration is selected, similar to KRUM, based on the smallest score. %We provide a detailed example for this in Appendix~\ref{}.

%\noindent\textbf{\sysname-enabled KRUM (Figure~\ref{fig:modular-fhe-krum}).}
Each aggregation server \( k \) receives a masked share
\( c_{i,k} = \Encop_{\mathsf{pk}}(r \cdot u_{i,k}) \) for every client
\( i \). For every ordered pair \((i,j)\), it homomorphically
computes the encrypted squared Euclidean distance of the corresponding
halves via PSE,
\[
  \Encop_{\mathsf{pk}}\!\bigl(d_{ij,k}\bigr)
  = \bigl(c_{i,k}-c_{j,k}\bigr)^{2}.
\]
Collaborative decryption exposes only the obfuscated value
\( d_{ij,k}=r^{2}(u_{i,k}-u_{j,k})^{2} \), after which
server \( k \) broadcasts a pairing commitment
\( e(g_0,g_0)^{\,d_{ij,k_\ell}} \) for every coordinate \(\ell\).

Let 
\(
\Delta_{ij,k_\ell}=u_{i,k_\ell}-u_{j,k_\ell}
\)
and recall that the commitment for a single coordinate is  
\(h_{i,k_\ell}=g^{u_{i,k_\ell}}=g_0^{r \cdot u_{i,k_\ell}}\).
Using bilinearity of the pairing \(e(\cdot,\cdot)\) we obtain
\[
\begin{aligned}
e\bigl(h_{i,k_\ell}h_{j,k_\ell}^{-1},\; h_{i,k_\ell}h_{j,k_\ell}^{-1}\bigr)
  &= e\bigl(g^{\Delta_{ij,k_\ell}},\; g^{\Delta_{ij,k_\ell}}\bigr) \\[-2pt]
  &\hspace{1em}\text{\small (substituting } h_{x,k_\ell} = g^{u_{x,k_\ell}} \text{)} \\[4pt]
  &= e(g, g)^{\Delta_{ij,k_\ell} \cdot \Delta_{ij,k_\ell}} \\[-2pt]
  &\hspace{1em}\text{\small (bilinearity: } e(g^a, g^b) = e(g, g)^{ab} \text{)} \\[4pt]
  &= e(g_0, g_0)^{r^2 \cdot (u_{i,k_\ell} - u_{j,k_\ell})^2} \\[2pt]
  &= e(g_0, g_0)^{d_{ij,k_\ell}}.
\end{aligned}
\]

Since the protocol establishes this relation,  
\(
e(g_0,g_0)^{r^2 \cdot (u_{i,k_\ell}-u_{j,k_\ell})^{2}}
           =e(g_0,g_0)^{d_{ij,k_\ell}}. 
\)
Because $e(g_0,g_0)$ generates $\mathbb G_T$, exponents coincide mod~$p$, yielding the stated equality.  
\[
d_{ij,k_\ell}\;=\;r^2 \cdot (u_{i,k_\ell}-u_{j,k_\ell})^{2}\bmod p .
\]
The peer server verifies these commitments with the above test.  The passing
check guarantees that server \( k \) has performed the aggregation
correctly.
Once all checks pass, either server (at least
one is assumed to be honest) 
computes the $n\!-\!f\!-\!2$ smallest (scaled) pairwise distances for every
client, sums them to obtain the Krum score, and selects the update with the minimal score as the next global model (using SCO and SF). 
Because the
mask preserves distance ordering, the plaintext post-processing is
identical to the classical Krum rule, while preserving privacy. % during the expensive arithmetic phase.

\subsubsection{\sysname-enabled Trimmed-Mean (Figure~\ref{fig:modular-fhe-tm})}

\begin{figure}[!htbp]
\setlength{\fboxsep}{.9pt}
\begin{center}
    \begin{tcolorbox}[enhanced,
    drop fuzzy shadow southwest, colframe=black,colback=white]
\footnotesize

\noindent\textbf{Input:} Each server $k$ receives a set of $n$ encrypted client updates $\{c_{i,k} = \Enc{r\cdot u_{i,k}}: i=1,2,\dots,n\}$, where $u_{i,k}$ is $k^{th}$ subset of client update of client \textit{i}. %$f$ is the bound on the number of Byzantine clients.

\noindent\rule{\textwidth}{0.1pt}

\noindent\textbf{Server-Side Encrypted Operations.}
\begin{enumerate}
    \item [] \textbf{Encrypted operations}
    \item For every ordered pair \((i,j)\) the servers compute,
\[
\Enc{d_{ij_{k}}} = (c_{i,k}-c_{j,k}), \quad \forall j \ne i.
\]
Note that,
\[
\Enc{d_{ij_{k}}} =\Enc{r \cdot u_{i,k}-r \cdot u_{j,k}}.
\]
    \item \textit{Collaborative decryption.}
    Decrypt to obtain
    
        \(d_{ij_k}=r \cdot (u_{i,k}-u_{j,k})\).

    \item []\textbf{Verification of Correct Computation.}
    
    \item Server publishes $g_0^{d_{ij,k_l}}$ as a proof for correct computation (for every coordinate \(\ell\)).
    \item Peer Server: For every coordinate \(\ell\) and pair \((i,j)\) check,
\[
  h_{ik_{\ell}} \cdot (h_{jk_{\ell}})^{-1} \mod p \stackrel{?}{=} g_0^{d_{ij,k_l}} \mod p.
\]    
    \noindent\rule{0.9\textwidth}{0.1pt}
    \item [] \textbf{Plaintext operations}

    \item \textbf{Coordinate-wise sorting:} The sign of each coordinate of $d_{ij_k}$ represents the order of coordinates at that position in two updates. 
    
    \begin{enumerate}
        \item From $d_{ij_k}$, we can prepare a binary sign vector $sgn$, such that $sgn_l = \begin{cases}
    0,& \text{if } u_{i,k_l} \leq u_{j,k_l}\\
    1,              & \text{otherwise}
\end{cases}$. 

        \item Then, prepare a $\textbf{1}$ vector containing $1$ in each coordinate i.e. $\textbf{1} = (1,1,\dots,1)$.

        \item Create temporary vectors as below: 
        \[
        \begin{aligned}
        \Enc{tmp_1} = sgn\cdot\Enc{u_{j k}} 
             \\ + (\textbf{1} - sgn)\cdot\Enc{u_{i k}}\,
        \end{aligned}
        \]
        
        \[
        \begin{aligned}
        \Enc{tmp_2} = sgn\cdot\Enc{u_{i k}}\\ 
             + (\textbf{1} - sgn)\,\Enc{u_{j k}}\,.
        \end{aligned}
        \]

        %\[\Enc{tmp_1} = (\textbf{1} - sgn) \cdot \Enc{u_{ik}} + sgn \cdot \Enc{u_{jk}} ,\]
        %\[\Enc{tmp_2} = sgn \cdot \Enc{u_{ik}} + (\textbf{1} - sgn) \cdot \Enc{u_{jk}} .\]

        \item $\Enc{tmp_1}$ and $\Enc{tmp_2}$ are sorted vectors, so we can assign them back to $\Enc{u_{ik}}$ and $\Enc{u_{jk}}$, respectively. This process repeats until the coordinates from all local model updates are sorted. %It uses the \textit{ Batcher’s odd-even mergesort network} network to generate compare-and-swap index pairs for efficient sorting.
    \end{enumerate}

    \item \textbf{Trim extremes:} This step is computed as for plaintext, pruning the top and bottom $f$ updates.
    \item \textbf{Selection:} FedAvg can be applied on the selected updates.
    
\end{enumerate}

\end{tcolorbox}
\end{center}
\caption{Trimmed-Mean in \sysname.} 
\label{fig:modular-fhe-tm}
\end{figure}

% After receiving the local updates from clients, servers perform encrypted operations in steps 1-3, and then, after doing obfuscation, they decrypt to get the plaintext shares. Then servers do coordinate-wise sorting for the updates in a pairwise fashion. It involves some plaintext-ciphertext computations. Finally, the extremes are trimmed, and the remaining updates are averaged to get the global model for the next iteration. %We provide a detailed example for this in Appendix~\ref{}.

%\noindent\textbf{\sysname-enabled Trimmed-Mean (Figure~\ref{fig:modular-fhe-tm}).}
Upon receiving the masked ciphertexts
\(c_{i,k}=\Encop_{\mathsf{pk}}(r\,u_{i,k})\) for every client \(i\),
each server \(k\) homomorphically subtracts every ordered pair \((i,j)\) to obtain the encrypted coordinate-wise differences
\(\Encop_{\mathsf{pk}}(d_{ij,k})\!=\!(c_{i,k}-c_{j,k})\), via PSE.
A collaborative decryption then reveals only the obfuscated values
\(d_{ij,k}=r\,(u_{i,k}-u_{j,k})\).
For each coordinate \(\ell\), the honest peer checks the relation
\(h_{i,k_\ell}h_{j,k_\ell}^{-1}\equiv g_0^{d_{ij,k_\ell}}\pmod p\),
thus ensuring the correctness of every difference without exposing plaintext updates.
A binary sign vector derived from the decrypted \(d_{ij,k}\) guides a homomorphic rearrangement that sorts the encrypted coordinates in ascending order (SCO).
Servers perform a secure compare-and-swap procedure, realised by Batcher’s odd-even mergesort network, on the ciphertexts. Batcher’s odd–even mergesort network executes a data‐independent sequence of $\mathcal{O}(n\log^{2}n)$ compare‐and‐swap operations, in contrast to the $\mathcal{O}(n^2)$ pairwise comparisons required by naive sorting. By integrating this network into Trimmed Mean, servers can efficiently perform coordinate-wise sorting over encrypted client updates, reducing the number of comparisons and overall computation compared to a fully quadratic approach.
Once sorting is complete,
the largest and smallest \(f\) values are discarded (SF), and the remaining
\(n-2f\) entries per dimension are averaged.
Because the random mask \(r\) preserves relative ordering, this
plaintext post-processing exactly reproduces the classical
The trimmed-mean rule, where all expensive computations, i.e., sorting, are done in the combined plaintext-ciphertext domain.
%thus retaining privacy without sacrificing robustness.

\subsubsection{\sysname-enabled FLTrust (Figure~\ref{fig:modular-fhe-fltrust})}

\begin{figure}[t!]
\setlength{\fboxsep}{.9pt}
\begin{center}
\begin{tcolorbox}[enhanced,
   drop fuzzy shadow southwest, colframe=black,colback=white]
\footnotesize

\noindent\textbf{Input:} 
\begin{itemize}
    \item Each server $k$ receives a set of $n$ encrypted client updates $\{c_{i,k} = \Enc{r\cdot u_{i,k}}: i=1,2,\dots,n\}$, where $u_{i,k}$ is $k^{th}$ subset of client update of client \textit{i}. %$f$ is the bound on the number of Byzantine clients. 
    \item Additionally receives encrypted norm $\{e_{i} = \Enc{\frac{r}{\| u_{i}\|}}: i=1,2,\dots,n\}$ [Note - It can also be calculated privately at the server side.]
    \item Server Update Vector: $u_0, \|u_0\|$
\end{itemize}
\noindent\rule{\textwidth}{0.1pt}

\noindent\textbf{Server-Side Encrypted Operations.}
\begin{enumerate}

\item \textit{Coordinate-wise products.}
For each client \(i\) 
\[
 \Encop_{\mathsf{pk}}(d_i) \;=\;
  u_{0} \odot c_{i,k} = \Encop_{\mathsf{pk}}(r \cdot u_{i,k} \odot u_0)
\]
\item \textit{Collaborative decryption.}
Decrypt to obtain $d_i$

\item []\textbf{Verification of Correct Computation.}
  \item Server Publishes $g_0^{d_{i_l}}$ as a proof for the correct computation (for every coordinate \(\ell\)).
  \item Peer Server: For every \(\ell\) check
        $$h_{i,k_l}^{u_{t_l}} \stackrel{?}{=}  g^{d_{i_l}}.$$
  \item Abort on the first mismatch.

\noindent\rule{0.9\textwidth}{0.1pt}
\noindent\textbf{Plaintext operations}
\item \textbf{Trust score.}
\[
  w_i=\max \left \{0,\quad \frac{\sum_{\ell=1}^{N} d_{i,\ell}}
            {\|u_0\|} \right\}\; .
\]

\item \textbf{Scaled Updates}
\[
 \Encop_{\mathsf{pk}}( u_i^{\text{scaled}}) 
 = \Encop_{\mathsf{pk}} \left ( w_i \cdot \frac{\|u_0\|}{\|u_i\|} \cdot u_i \right),
\]
\[
= \left (w_i \cdot\|u_0\|  \right) \cdot \Encop_{\mathsf{pk}}\left(\frac{1}{\|u_i\|}\right)\cdot \Encop_{\mathsf{pk}}(u_i).
\]
\item \textbf{Weighted Aggregation.}\;
\[
\Encop_{\mathsf{pk}}(u^{t+1} )= \frac{1}{\sum_{i=1}^{n} w_i} \cdot \sum_{i=1}^{n} {\Encop_\mathsf{pk}(u_i^{\text{scaled}})}
\]
\end{enumerate}

\end{tcolorbox}
\end{center}
\caption{FLTrust in \sysname.}
\label{fig:modular-fhe-fltrust}
\end{figure}

%\noindent\textbf{\sysname-enabled FLTrust (Figure~\ref{fig:modular-fhe-fltrust}).}
Every server \( \mathcal S_k \) receives for each client \( i \)
(i) a masked ciphertext share
\(c_{i,k}= \Encop_{\mathsf{pk}}(r\cdot u_{i,k})\)  
and (ii) the encrypted reciprocal norm
\(e_i=\Encop_{\mathsf{pk}}\!\bigl(\|u_i\|^{-1}\bigr)\);
the server itself holds the trusted reference vector
\(u_0\) and its norm \(\|u_0\|\).
These computations rely on linear aggregation (LA). 
It first evaluates, under encryption, the coordinate-wise products
\( \Encop_{\mathsf{pk}}(d_i) 
  =\Encop_{\mathsf{pk}}\!\bigl(r\cdot u_{i,k}\odot u_0\bigr) \).
A threshold decryption reveals only the masked inner product
vector \(d_i=r\cdot u_{i,k}\odot u_0\).
For every coordinate \(\ell\) the peer server verifies
\(h_{i,k_\ell}^{\,u_{t,\ell}}\stackrel{?}{=}g^{d_{i,\ell}}\);
failure of any check aborts the protocol, ensuring integrity via the discrete-log assumption.
When all checks pass, the score vectors are concatenated to
form \(d_i\in\mathbb R^{N}\).
Each server then computes the client’s FLTrust weight
\(
  w_i=\max\{0,\sum_{\ell=1}^{N}d_{i,\ell}/\|u_0\|\},
\)
a quantity still proportional to the true cosine similarity.
Finally, the update is rescaled homomorphically (LA):
\[
  \Encop_{\mathsf{pk}}(u_i^{\text{scaled}})=
  \Encop_{\mathsf{pk}}\!\bigl(
      w_i\,\tfrac{\|u_0\|}{\|u_i\|}\,u_i
  \bigr),
\]
and the weighted average
\(
  \Encop_{\mathsf{pk}}(u^{t+1})=
  \bigl(\sum_i w_i\bigr)^{-1}\!
  \sum_i \Encop_{\mathsf{pk}}\!\bigl(u_i^{\text{scaled}}\bigr)
\)
is produced under encryption.
%Because all costly arithmetic occurs on masked ciphertexts, privacy of individual updates is preserved, while the plaintext post-processing remains identical to the classical FLTrust aggregation rule.

\subsubsection{Enabling other aggregation schemes}
Beyond the three schemes instantiated above, \sysname's modular design extends naturally to a broad class of Byzantine-robust aggregation schemes. 
We observe that defenses such as MESAS~\cite{krauss2023mesas}, CrowdGuard~\cite{rieger2022crowdguard}, and FLAME~\cite{nguyen2022flame} can be tailored to \sysname.
FLAME computes cosine similarity and performs clustering after that (LA, PSE). 
Further, it computes the Euclidean distance and finds the median for adaptive clipping (SCO). As discussed above, most of these statistical operations can be tailored to \sysname. CrowdGuard also computes layer-wise Cosine and Euclidean distances, which we have already demonstrated with Krum and FLTrust. MESAS includes six statistical operations, including cosine similarity, Euclidean distance, count, variance, maximum, and minimum (PSE). 
Subsequently, it performs multiple clustering-based statistical tests to prune (LA, SF). 
Since MESAS covers the full set of statistical operations, we describe how to tailor it to \sysname. Details can be referred to in~\ref{mesas}. %Security analysis is discussed in Appendix~\ref{app1} -~\ref{app4}.

%-------------------------------------------------
\section{Security Analysis}\label{sec:security}
%-------------------------------------------------
This section establishes that the proposed protocol guarantees confidentiality and integrity against both honest-but-curious and actively malicious aggregation servers.  All reductions invoke the \emph{Decisional Ring–LWE} assumption for ciphertext privacy and the hardness of the
\emph{Discrete Logarithm Problem} (DLP) in a prime-order group for commitment soundness.

Throughout, let~$\mathcal S_{1},\mathcal S_{2}$ be the two aggregation
servers, $\mathsf{pk}$ the public multi-key FHE key, and let
\[
  \mathsf{view}^{\mathrm{pre}}_k
  \;=\;
  \bigl\{\,
      c_{i,k},\,
      h_{i,k_\ell},\,
      e_i
  \bigr\}_{i,\ell}
\]
denote the entire information set available to
server~$\mathcal S_k$ \emph{before} any homomorphic processing
($c_{i,k}=\Encop_{\mathsf{pk}}(r \cdot u_{i,k})$,
$h_{i,k_\ell}=g^{u_{i,k_\ell}}$,
and, if required, $e_i=\Encop_{\mathsf{pk}}(\|u_i\|^{-1})$).

Let
\[
  \mathsf{view}^{\mathrm{post}}_k
  \;=\;
  \mathsf{view}^{\mathrm{pre}}_k
   \cup
  \bigl\{\,
      d_i  \mid  i\in[n]
  \bigr\}
\]
be the information set \emph{after} collaborative decryption, where $d_i$ is either the masked pairwise‐distance vector (Krum), the masked coordinate difference vector (Trimmed‐Mean), or the masked inner‐product vector (FLTrust).

Further, each server publishes a proof element \(g^{d_i}\in\mathbb G\) and extends their local view to
\[
  \mathsf{view}^{\mathrm{proof}}_k
  =\{g^{d_{i_k}}\mid i\in[n]\}.
\]
%Finally, let $\sigma$ be the final plaintext aggregate or selection released to all parties.

Each peer server validates the published
\(g^{d_i}\) against client commitments
\(h_{i,k_\ell}\) using the method-specific equation:
\[
\text{Krum:}\;
  e\!\bigl(h_{i,k_\ell}h_{j,k_\ell}^{-1},
           h_{i,k_\ell}h_{j,k_\ell}^{-1}\bigr)
  \stackrel{?}{=} e(g,g)^{d_{ij,k_\ell}}.
\]
\[
\text{Trimmed-Mean:}\;
  h_{i,k_\ell}h_{j,k_\ell}^{-1}\stackrel{?}{=}g^{d_{ij,k_\ell}}.
\]
\[
\text{FLTrust:}\;
  h_{i,k_\ell}^{u_{t,\ell}}\stackrel{?}{=}g^{d_{i,\ell}}.
\]
A single failure causes the honest server to withhold sharing a decrypted share, forcing an abort and preventing acceptance of the forged data.

The sequence of Theorems~\ref{thm:input-conf} to~\ref{thm:integrity} discussed in~\ref{app1} to ~\ref{app3} uses the above views
\(\mathsf{view}^{\mathrm{pre}}_k\),
\(\mathsf{view}^{\mathrm{post}}_k\), and 
\(\mathsf{view}^{\mathrm{proof}}_k\) to prove confidentiality and integrity for the
entire protocol.

\section{Evaluation}
In this section, we present a comprehensive evaluation of \sysname, evaluating its computation, communication costs, robustness and verification cost analysis.
%, our proposed privacy-preserving, Byzantine-resilient and verifiable federated learning (FL) framework. 
%Our goal is to demonstrate that \sysname achieves practical performance while offering strong confidentiality, robustness, and verifiability guarantees,  improving over existing secure FL approaches.
%in terms of both scalability and efficiency.

\subsection{Experimental setup}
\subsubsection{Implementation} All experiments are conducted on a system with Intel Xeon E5-4650 v4 @ 2.20 GHz with 112 logical CPUs, and 256 GB RAM. We implemented our framework in C++ using the Microsoft SEAL library\footnote{https://github.com/microsoft/SEAL} for homomorphic encryption. 
Our code is available here\footnote{https://github.com/harshkasyap/provefl}. 
For local model training, we used the PyTorch\footnote{https://github.com/pytorch/pytorch} framework. 
We instantiated clients to run independently. 
Multi-key support and ciphertext packing are leveraged to maximize throughput. 
All results reported are averaged over 5 independent runs, unless specified otherwise.

\subsubsection{Evaluation Goals} Our evaluation seeks to answer the following key questions:

\paragraph*{Efficiency.} What is the computational and communication overhead introduced by encrypted computation and partial decryption, to integrate state-of-the-art Byzantine robust schemes? How does \sysname scale with increasing number of clients, servers, and model size?

\paragraph*{Breakdown.} What is the breakdown of individual components in \sysname? We chose to show this to demonstrate the time savings from offloading computations to plaintext.

\paragraph*{Byzantine Robustness.} How well does the system perform (in terms of accuracy) in the presence of Byzantine clients under various poisoning attacks (Backdoor Attack, Trim Attack) and state-of-the-art Byzantine robust schemes tailored to \sysname?

\paragraph*{Comparative Performance.} How does \sysname compare to prior secure FL systems such as ELSA~\cite{rathee2023elsa}, Prio~\cite{addanki2022prio+,corrigan2017prio} and RoFL~\cite{burkhalter2021rofl}? We chose these systems to compare with, since they cover single-server and distributed servers and can support at least a Byzantine-robust scheme (norm clipping).

\paragraph*{Verification Cost.} We analsze the verification cost incurred by the server to validate that the aggregation has been performed correctly, assuming that at least one server is honest.

\subsubsection{Parameters and Variations} We vary the following system parameters to study performance under different settings:

-- Number of Clients (C): {10, 50, 100, 200}.

-- Number of Servers (S): {2, 4, 10}.

-- Dataset - Model Size (P): (1) 100k and 500k for Benchmarking; (2) 62k (CIFAR-10 S~\cite{krizhevsky2009learning} - LeNet5~\cite{lecun1989backpropagation}); (3) 273k (CIFAR-10 L~\cite{krizhevsky2009learning} - ResNet-18~\cite{he2016deep}); (4) 818k (Shakespeare~\cite{caldas2018leaf} - LSTM~\cite{hochreiter1997long})
    
-- Poisoning Attacks: Trim~\cite{fang2020local} and Backdoor~\cite{bagdasaryan2020backdoor} attacks.

-- Poisoning Client Fraction (A\%): {10\%, 20\%, 40\%.}
    
-- Byzantine-Robust Aggregation Rules: $L_2$ norm (for comparison), Krum~\cite{blanchard2017machine}, Trimmed Mean~\cite{yin2018byzantine}, FLTrust~\cite{cao2020fltrust}.

-- Baselines: ELSA~\cite{rathee2023elsa}, RoFL~\cite{burkhalter2021rofl}, Prio~\cite{corrigan2017prio,addanki2022prio+}.

-- FHE polynomial modulus degree - 32768.

\begin{figure*}[t!]
\centering
%\hfill
\begin{subfigure}{0.24\textwidth}
\centering
\begin{tikzpicture}
\begin{axis}[
    width=\textwidth,
    height=4cm,
    xlabel={\# Clients},
    ylabel={Time (s)},
    title={L2 Norm},
    legend pos=north west,
    legend style={font=\scriptsize},
    xtick={10,50,100,200},
    ymin=0,
    grid=major,
]
\addplot[mark=*, blue] coordinates {(10, 0.196416) (50, 0.762095) (100, 1.46639) (200, 2.90938)};
\addplot[mark=square*, red] coordinates {(10, 0.160299) (50, 0.697062) (100, 1.35593) (200, 2.67682)};
\addplot[mark=triangle*, green!60!black] coordinates {(10, 0.142029) (50, 0.652917) (100, 1.31738) (200, 2.56295)};
\legend{2 servers, 4 servers, 10 servers}
\end{axis}
\end{tikzpicture}
\end{subfigure}%
%\hfill
\begin{subfigure}{0.24\textwidth}
\centering
\begin{tikzpicture}
\begin{axis}[
    width=\textwidth,
    height=4cm,
    xlabel={\# Clients},
    title={Krum},
    xtick={10,50,100,200},
    ymin=0,
    grid=major,
]
\addplot[mark=*, blue] coordinates {(10, 0.518932) (50, 9.89996) (100, 164.501) (200, 932.087)};
\addplot[mark=square*, red] coordinates {(10, 0.449708) (50, 7.90608) (100, 154.053) (200, 662.751)};
\addplot[mark=triangle*, green!60!black] coordinates {(10, 0.355812) (50, 7.73843) (100, 123.761) (200, 543.251)};
\end{axis}
\end{tikzpicture}
\end{subfigure}%
%\hfill
\begin{subfigure}{0.24\textwidth}
\centering
\begin{tikzpicture}
\begin{axis}[
    width=\textwidth,
    height=4cm,
    xlabel={\# Clients},
    title={Trimmed Mean},
    xtick={10,50,100,200},
    ymin=0,
    grid=major,
]
\addplot[mark=*, blue] coordinates {(10, 4.2624183) (50, 55.2136228) (100, 161.601243) (200, 404.513319)};
\addplot[mark=square*, red] coordinates {(10, 2.2161577) (50, 28.9768077) (100, 81.300234) (200, 202.281614)};
\addplot[mark=triangle*, green!60!black] coordinates {(10, 1.0882726) (50, 14.3614979) (100, 38.8895205) (200, 102.6255565)};
\end{axis}
\end{tikzpicture}
\end{subfigure}%
%\hfill
\begin{subfigure}{0.24\textwidth}
\centering
\begin{tikzpicture}
\begin{axis}[
    width=\textwidth,
    height=4cm,
    xlabel={\# Clients},
    title={FLTrust},
    xtick={10,50,100,200},
    ymin=0,
    grid=major,
]
\addplot[mark=*, blue] coordinates {(10, 0.932282) (50, 4.53293) (100, 9.02104) (200, 18.3548)};
\addplot[mark=square*, red] coordinates {(10, 0.616188) (50, 3.02731) (100, 6.09845) (200, 12.1634)};
\addplot[mark=triangle*, green!60!black] coordinates {(10, 0.46568) (50, 2.28502) (100, 4.53496) (200, 9.0667)};
\end{axis}
\end{tikzpicture}
\end{subfigure}

\caption{Computation time vs. no. of clients (and servers) for different Byzantine-robust aggregation rules (100k params).}
\label{fig:aggregation-times-1}
\end{figure*}
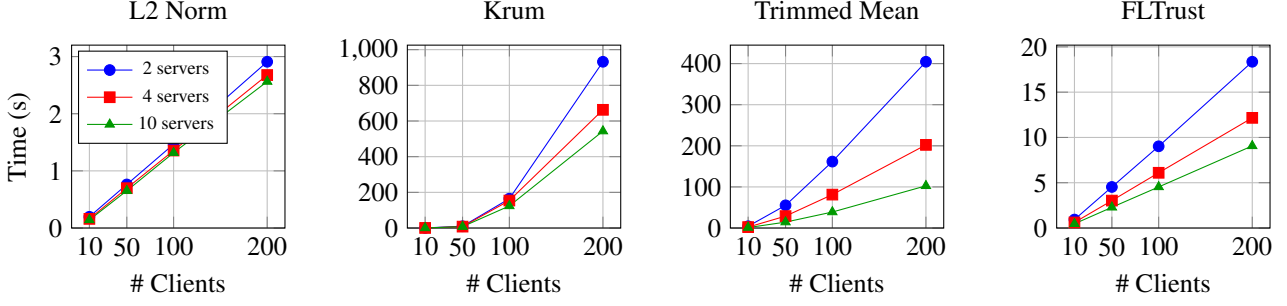

\begin{figure*}[t!]
\centering
% $L_2$ norm Plot
\begin{subfigure}{0.24\textwidth}
\centering
\begin{tikzpicture}

\begin{axis}[
    width=\textwidth,
    height=4cm,
    xlabel={\# Clients},
    ylabel={Time (s)},
    title={L2 Norm},
    legend pos=north west,
    legend style={font=\scriptsize},
    xtick={10,50,100,200},
    ymin=0,
    grid=major,
]
\addplot+[mark=*, blue] coordinates {(10,0.140929) (50,0.655231) (100,1.28545) (200,2.55964)};
\addlegendentry{62k}
\addplot+[mark=square*, red] coordinates {(10,0.213762) (50,0.849987) (100,1.51453) (200,2.96031)};
\addlegendentry{273k}
%\addplot+[mark=square*] coordinates {(10,0.269433) (50,0.931466) (100,1.69879) (200,3.40037)};
%\addlegendentry{500k}
\addplot+[mark=triangle*, green!60!black] coordinates {(10,0.376273) (50,1.05894) (100,2.08115) (200,4.03572)};
\addlegendentry{818k}
\end{axis}
\end{tikzpicture}
\end{subfigure}
% Krum Plot
\begin{subfigure}{0.24\textwidth}
\centering
\begin{tikzpicture}
\begin{axis}[
    width=\textwidth,
    height=4cm,
    title={Krum},
    xtick={10,50,100,200},
    xlabel={\#Clients},
    ymin=0,
    grid=major
]
\addplot+[mark=*, blue] coordinates {(10,0.356089) (50,6.16055) (100,123.761) (200,543.251)};
%\addlegendentry{62k}
\addplot+[mark=square*, red] coordinates {(10,0.552175) (50,11.5088) (100,169.603) (200,1059.91)};
%\addlegendentry{273k}
%\addplot+[mark=square*] coordinates {(10,0.70652) (50,18.0251) (100,190.587) (200,1425.94)};
%\addlegendentry{500k}
\addplot+[mark=triangle*, green!60!black] coordinates {(10,0.933462) (50,21.7743) (100,217.85) (200,2122.22)};
%\addlegendentry{818k}
\end{axis}
\end{tikzpicture}
\end{subfigure}
% Trimmed-Mean Plot
\begin{subfigure}{0.24\textwidth}
\centering
\begin{tikzpicture}
\begin{axis}[
    width=\textwidth,
    height=4cm,
    title={Trimmed-Mean},
    xtick={10,50,100,200},
    xlabel={\#Clients},
    ymin=0,
    grid=major
]
\addplot+[mark=*, blue] coordinates {(10,0.5889662) (50,7.7693752) (100,21.3936372) (200,56.2604686)};
%\addlegendentry{62k}
\addplot+[mark=square*, red] coordinates {(10,2.923016) (50,38.97521) (100,106.937912) (200,296.113413)};
%\addlegendentry{273k}
%\addplot+[mark=square*] coordinates {(10,4.70586) (50,63.892556) (100,174.600615) (200,477.386976)};
%\addlegendentry{500k}
\addplot+[mark=triangle*, green!60!black] coordinates {(10,7.598931) (50,103.78951) (100,287.165813) (200,616.957911)};
%\addlegendentry{818k}
\end{axis}
\end{tikzpicture}
\end{subfigure}
% FLTrust Plot
\begin{subfigure}{0.24\textwidth}
\centering
\begin{tikzpicture}
\begin{axis}[
    width=\textwidth,
    height=4cm,
    title={FLTrust},
    xtick={10,50,100,200},
    xlabel={\#Clients},
    ymin=0,
    grid=major
]

\addplot+[mark=*, blue] coordinates {(10,0.459768) (50,2.2807) (100,4.51924) (200,9.0213)};
%\addlegendentry{62k}
\addplot+[mark=square*, red] coordinates {(10,1.06813) (50,5.38836) (100,10.7192) (200,21.8529)};
%\addlegendentry{273k}
%\addplot+[mark=square*] coordinates {(10,1.51459) (50,7.66901) (100,15.23) (200,30.776)};
%\addlegendentry{500k}
\addplot+[mark=triangle*, green!60!black] coordinates {(10,2.36229) (50,11.7428) (100,23.5983) (200,46.7419)};
%\addlegendentry{818k}
\end{axis}
\end{tikzpicture}
\end{subfigure}

\caption{Computation time vs. no. of clients (with 2 servers) for different Byzantine-robust aggregation rules across \textit{different} model sizes (and datasets). 62k (CIFAR-10 S - LeNet5), 273k (CIFAR-10 L - ResNet-18), 818k (Shakespeare - LSTM).}
\label{fig:agg-2}
\end{figure*}
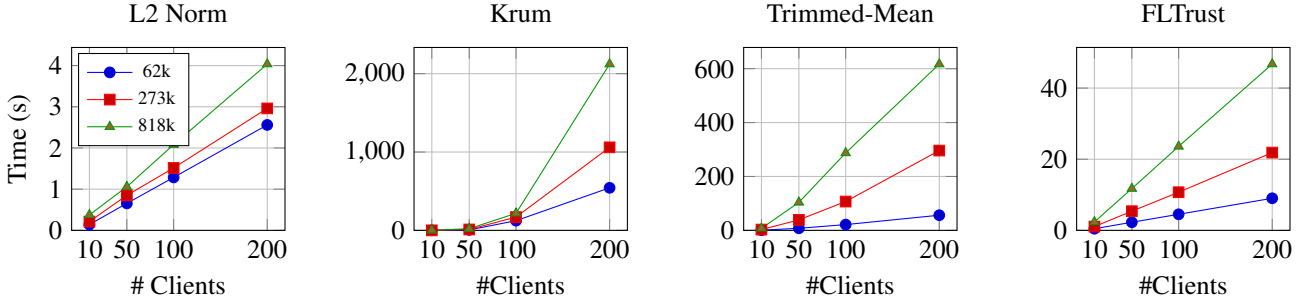

\begin{table}[b]
\centering
\caption{Communication cost (in GB) vs. no. of clients with 2 servers (100k Params).}
\label{tab:comm-cost-2s}
\resizebox{0.5\textwidth}{!}{%
\begin{tabular}{|c|c|c|c|c|}
\hline
\textbf{\# Clients} & \textbf{$L_2$ norm} & \textbf{Krum} & \textbf{Trimmed-Mean} & \textbf{FLTrust} \\
\hline
10   & 0.005 & 0.110  & 0.061 & 0.026 \\\hline
50   & 0.005 & 3.220  & 0.824 & 0.132 \\\hline
100  & 0.005 & 13.030 & 2.257 & 0.263 \\\hline
200  & 0.005 & 26.180 & 5.924 & 0.527 \\\hline
\multicolumn{5}{l}{Only $L_2$ norm values change with increasing number of servers.}\\\hline
\end{tabular}
}
\end{table}

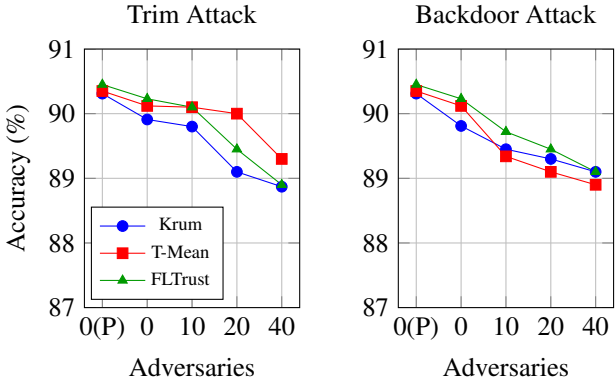
\begin{figure}[t!]
\centering

% --- Trim Attack Subfigure ---
\begin{subfigure}[t]{0.5\columnwidth}
\centering
\begin{tikzpicture}
\begin{axis}[
    width=\columnwidth,
    height=5cm,
    xlabel={Adversaries},
    ylabel={Accuracy (\%)},
    title={Trim Attack},
    xtick={1,2,3,4,5},
    xticklabels={0(P), 0, 10, 20, 40},
    ymin=87, ymax=91,
    legend pos=south west,
    legend style={font=\scriptsize},
    grid=major
]
\addplot[mark=*, blue] coordinates {
    (1, 90.31) (2, 89.91) (3, 89.80) (4, 89.10) (5, 88.87)
};
\addlegendentry{Krum}

\addplot[mark=square*, red] coordinates {
    (1, 90.35) (2, 90.12) (3, 90.10) (4, 90.00) (5, 89.30)
};
\addlegendentry{T-Mean}

\addplot[mark=triangle*, green!60!black] coordinates {
    (1, 90.45) (2, 90.23) (3, 90.10) (4, 89.45) (5, 88.90)
};
\addlegendentry{FLTrust}
\end{axis}
\end{tikzpicture}
\end{subfigure}%
\hfill
% --- Backdoor Attack Subfigure ---
\begin{subfigure}[t]{0.5\columnwidth}
\centering
\begin{tikzpicture}
\begin{axis}[
    width=\columnwidth,
    height=5cm,
    xlabel={Adversaries},
    title={Backdoor Attack},
    xtick={1,2,3,4,5},
    xticklabels={0(P), 0, 10, 20, 40},
    ymin=87, ymax=91,
    legend pos=south west,
    legend style={font=\scriptsize},
    grid=major
]
\addplot[mark=*, blue] coordinates {
    (1, 90.31) (2, 89.81) (3, 89.45) (4, 89.30) (5, 89.10)
};
%\addlegendentry{Krum}

\addplot[mark=square*, red] coordinates {
    (1, 90.35) (2, 90.12) (3, 89.34) (4, 89.10) (5, 88.90)
};
%\addlegendentry{T-Mean}

\addplot[mark=triangle*, green!60!black] coordinates {
    (1, 90.45) (2, 90.23) (3, 89.72) (4, 89.45) (5, 89.10)
};
%\addlegendentry{FLTrust}
\end{axis}
\end{tikzpicture}
\end{subfigure}

\caption{Accuracy under Trim and Backdoor Attacks for different aggregation rules (CIFAR-10 L - ResNet-18). P: Plaintext, Rest: Encrypted and simulated with 2 servers.}
\label{fig:attack-comparison}
\end{figure}

\begin{table*}[t!]
\centering
\caption{Computation and Communication Cost (in MB) for Different Methods}
\label{tab:comp-cost}
\resizebox{\textwidth}{!}{%
\begin{tabular}{|c|c|cc|cc|cc|cc|}
\hline
\textbf{\# Clients} & \textbf{\# Params} &
\multicolumn{2}{c|}{\textbf{Prio}} &
\multicolumn{2}{c|}{\textbf{ELSA}} &
\multicolumn{2}{c|}{\textbf{\sysname (2 Servers)}} &
\multicolumn{2}{c|}{\textbf{\sysname (4 Servers)}} \\
\cline{3-10}
& & \textbf{Client} & \textbf{Server} & \textbf{Client} & \textbf{Server} & \textbf{Client} & \textbf{Server} & \textbf{Client} & \textbf{Server} \\
\hline
50  & 100k & 14.3 (59.1)  & 23.3 (0.002)  & 4.6 (51.6)  & 2.7 (640)    & 0.182 (9.8)  & 0.76 (5)  & 0.188 (10)  & 0.71 (2.5) \\
100 & 100k & 14.8 (59.1)  & 48.9 (0.005)  & 7.1 (51.6)  & 3.8 (1280)   & 0.182 (10)  & 1.46 (5)  & 0.188 (10)  & 1.36 (2.5) \\
200 & 100k & 16.5 (59.1)  & 99.5 (0.010)  & 8.4 (51.6)  & 6.1 (2560)   & 0.182 (10)  & 2.90 (5)  & 0.188 (10)  & 2.68 (2.5) \\
\hline
50  & 500k & 63.6 (262.2) & 102.9 (0.002) & 17.5 (258.0) & 11.2 (3200) & 0.730 (20)  & 1.01 (10) & 0.742 (20)  & 0.93 (10) \\
100 & 500k & 67.7 (262.2) & 218.4 (0.005) & 23.2 (258.0) & 17.3 (6400) & 0.730 (20)  & 1.79 (10) & 0.742 (20)  & 1.69 (10) \\
200 & 500k & 78.3 (262.2) & 457.9 (0.010) & 38.0 (258.0) & 31.4 (12800) & 0.730 (20) & 3.52 (10) & 0.742 (20)  & 3.41 (10) \\
\hline
\multicolumn{10}{l}{Comparison of runtime (sec) and data sent (MB in parentheses; per client and per server) for relaxed $L_\infty$ defense.}\\\hline
\end{tabular}
}
\end{table*}

\begin{table}[t!]
\centering
\caption{Computation time (in seconds) and communication (in GB) for different model sizes (and datasets).}
\resizebox{0.5\textwidth}{!}{%
\begin{tabular}{|c|cc|cc|cc|}
\hline
\textbf{\#Params} & \multicolumn{2}{c|}{\textbf{RoFL}} & \multicolumn{2}{c|}{\textbf{ELSA}} & \multicolumn{2}{c|}{\textbf{\sysname}} \\
\cline{2-7}
         & Time & Comm & Time & Comm & Time & Comm \\
\hline
62k      & 278  & 0.8  & 1.9  & 0.9  & 0.655 & 0.001 \\
273k     & 2229 & 3.8  & 7.3  & 7.3  & 0.850 & 0.006 \\
818k     & 4742 & 11.4 & 18.1 & 18.1 & 1.059 & 0.016 \\
\hline
\end{tabular}
}
\label{tab:comparison}
\end{table}

\subsection{Results}

\subsubsection{Efficiency} Figure~\ref{fig:aggregation-times-1} presents a comparative analysis of computation time across four Byzantine-robust aggregation rules: $L_2$ Norm, Krum, Trimmed Mean, and FLTrust, with varying numbers of clients (10, 50, 100, 200) and server counts (2, 4, 10).

We observe that $L_2$ Norm exhibits the lowest and most stable computation time across all configurations, with only a modest increase as the client count grows, and minimal variation with server scaling. Krum incurs a high computational cost, especially with larger client counts, due to its pairwise distance computations. This cost is alleviated by increasing the number of servers. Trimmed Mean also benefits notably from increased server parallelism, suggesting good scalability. However, the costs of Krum and Trimmed Mean are inherently quadratic in the number of clients, since they compare all pairs of updates, whereas the costs of L2 norm and FLTrust are linear in the number of clients.

Figure~\ref{fig:agg-2} demonstrates the server-side computation time as the number of clients increases across models of varying sizes (different datasets, two servers). The $L_2$ Norm method, being the simplest, exhibits linear scaling and remains computationally efficient even with 200 clients and large models. As noted earlier, Krum is significantly more expensive due to pairwise distance calculations, resulting in prohibitive delays for large-scale deployments. Trimmed-Mean offers an acceptable cost along with scalability. FLTrust maintains scalability similar to that of the $L_2$ Norm. The costs of Krum and Trimmed Mean are inherently quadratic in the number of clients, since they compare all pairs of updates, whereas the costs of $L_2$ norm and FLTrust are linear in the number of clients.

Table \ref{tab:comm-cost-2s} presents the communication cost (in GB) incurred under different Byzantine-robust aggregation rules as the number of clients increases, assuming two servers and a fixed model size of 100k parameters. $L_2$ norm shows a constant communication cost of 0.005 GB, regardless of the number of clients. This is because it requires only aggregated statistics to be shared with the servers, making it highly communication-efficient. Krum's communication cost grows to over 26 GB for 200 clients because each client must exchange partial distance computations with all others. Trimmed-Mean scales more moderately than Krum, as it requires fewer exchanges, since we utilized the index pairs from Batcher’s odd-even mergesort network. FLTrust has a significantly lower cost than Krum and Trimmed-Mean. We refer to~\ref{app5} for the runtime breakdown.

\subsubsection{Byzantine Robustness} Figure \ref{fig:attack-comparison} demonstrates that all three aggregation rules (Krum, Trimmed-Mean, and FLTrust) maintain high accuracy even when clients’ updates are encrypted and subjected to increasing proportions of adversarial clients (performing Trim~\cite{fang2020local} and Backdoor~\cite{bagdasaryan2020backdoor} attacks).
With no attack (both in plaintext and encrypted cases), each method achieves around 90\% accuracy. As the adversarial fraction increases to 40\%, the accuracy drop is minimal. 
These results confirm that, despite the added encryption overhead and targeted attacks, the robust aggregation schemes continue to behave as intended when implemented within \sysname.

\subsubsection{Comparative Performance} Table~\ref{tab:comp-cost} compares the computation time per client and per server (in seconds) and the communication cost (in MB) between different methods. We assume a relaxed $L_\infty$ defense for comparison, because ELSA does not support complex defenses, whereas Prio needs adjustments to run an $L_2$ norm defense as well. Prio incurs the highest overhead, with per-client costs exceeding 78s and 262MB for larger models (500k parameters), and server communication scaling linearly with the number of clients. ELSA improves client-side computation but suffers from excessive server communication and computation, particularly at larger scales (e.g., 31.4s and 12.5GB at 200 clients).

\sysname significantly reduces overhead for both clients and servers. In all configurations, client computation remains less than 1s, and communication cost is bounded (10–20MB), making it suitable for deployment even on edge devices. Server-side costs scale modestly with the number of clients and remain well within practical limits, while communication costs are held constant by efficient encrypted aggregation. Notably, using 4 servers instead of 2 further balances the load without increasing overall communication. These results demonstrate that \sysname achieves scalability and practical efficiency, outperforming existing approaches. 

Table~\ref{tab:comparison} evaluates how computation time and communication scale with increasing model size across RoFL, ELSA, and \sysname. RoFL exhibits extremely high computation costs (up to 4742s) and communication overheads, making it impractical for large models. ELSA significantly reduces computation time compared to RoFL, but still incurs high communication costs that scale linearly with model size. In contrast, \sysname maintains low and stable computation times, even for the largest model and drastically reduces communication overhead, up to 17 MB. This demonstrates that \sysname is far more efficient and scalable than existing solutions, and also supports complex aggregation rules.
%\smallskip

\subsubsection{Verification Cost Analysis} Let \(N\) be the full update dimension and \(n = N/2\) the number of coordinates held by each server. For every coordinate \(\ell \in \{1, \dots, n\}\), the honest peer checks the proof returned by the other server. Table~\ref{tab:verify-cost} counts, \emph{per coordinate}, the two dominant operations:

\begin{itemize}
  \item \(\mathsf{Exp}\): a modular exponentiation in the source group (\(\mathbb{G}\) or \(\mathbb{G}_T\));
  \item \(\mathsf{Pair}\): a single evaluation of the bilinear pairing \(e\colon \mathbb{G} \times \mathbb{G} \to \mathbb{G}_T\).
\end{itemize}

\begin{table}[h!]
\centering
\caption{Verification cost per coordinate (the honest peer performs the same test for all \(n = N/m\) coordinates).}
\label{tab:verify-cost}
\setlength{\tabcolsep}{6pt}
\renewcommand{\arraystretch}{1.15}
\resizebox{0.5\textwidth}{!}{%
\begin{tabular}{@{}|l|c|c|@{}}
\hline
\textbf{Method} & \#\,\(\mathsf{Exp}\) / coordinate & \#\,\(\mathsf{Pair}\) / coordinate\\
\hline
Krum         & \(1\) \quad (\(e(g,g)^d\))        & \(1\) \quad (\(e(\cdot,\cdot)\)) \\\hline
Trimmed-Mean & \(1\) \quad (\(g^d\))             & \(0\) \\\hline
FLTrust     & \(2\) \quad (\(h_i^{u_t},\,g^d\)) & \(0\) \\
\hline
\end{tabular}
}
\end{table}

\paragraph*{Aggregate cost} Multiplying the per-coordinate counts by \(n\) yields the total work per round:
\[
\begin{aligned}
\text{Krum:}         & \quad n\,\mathsf{Exp} + n\,\mathsf{Pair}, \\
\text{Trimmed-Mean:} & \quad n\,\mathsf{Exp}, \\
\text{FLTrust:}     & \quad 2n\,\mathsf{Exp}.
\end{aligned}
\]

Because a single pairing is typically an order of magnitude slower than a modular exponentiation, Krum incurs the highest overhead
%; Trimmed-Mean is the cheapest, requiring only one exponentiation per coordinate; FLTrust sits in between with two exponentiations but no pairing. 
In all three cases, the verification effort scales linearly with the model dimension and is independent of the number of clients, ensuring scalability to high-dimensional updates.

\subsubsection{Limitations} 
\sysname is compatible with a broad class of Byzantine-robust aggregation schemes, where relevant structural properties such as ordering, sign, or magnitude ratios are preserved under multiplicative masking and the operations can also be distributed among multiple servers to jointly calculate the final value. This holds for distance-based, similarity-based, and sorting-based operations. However, there do exist some aggregation schemes, such as Min-Max and Min-Sum~\cite{shejwalkar2021manipulating}, that apply non-linear transformations to client updates or rely on adaptive thresholds that are themselves functions of the unmasked data, and thus cannot be straightforwardly instantiated in \sysname.

\section{Conclusion}
This work presented a practically efficient and provably secure federated learning framework \sysname based on multiparty key homomorphic encryption utilised by multiple servers. It is designed to address the challenges of data confidentiality and data aggregation integrity in the presence of malicious participants and servers. By distributing encrypted computation across multiple servers and introducing collaborative randomization to obfuscate sensitive intermediate values, \sysname achieves scalability without compromising privacy. The framework supports different Byzantine-robust aggregation mechanisms, such as Krum, Trimmed Mean, FLTrust and MESAS directly on encrypted inputs, demonstrating resilience against both inference and poisoning attacks. %Extensive evaluations confirm that our design achieves comparable performance, even under strong adversarial conditions, while incurring minimal cryptographic overhead. 
The ability to efficiently integrate verification checks, traditionally impractical in pure FHE settings, further underscores the deployability of our solution. 
%This work provides a principled foundation for secure, trustworthy, and scalable federated learning in adversarial and privacy-sensitive environments.

\section*{CRediT authorship contribution statement} Harsh Kasyap, Anil Kumar Pradhan, and Ugur Ilker Atmaca are responsible for the design of the scheme, the experiment, and the writing of the paper. Graham Cormode and Carsten Maple are responsible for the design of the scheme and editing of the paper.

\section*{Declaration of competing interest} The authors declare that they have no known competing financial interests or personal relationships that could have appeared to influence the work reported in this paper.

\section*{Data availability} Public data has been used.

\section*{Acknowledgment} This work is supported, in part, by the UKRI Prosperity Partnership Scheme (FAIR) under the EPSRC Grant EP/V056883/1; EP/R007195/1 (Academic Centre of Excellence in Cyber Security Research - University of Warwick); EP/N510129/1 (The Alan Turing Institute); the Bill and Melinda Gates Foundation [INV-001309]. Under the grant conditions of the Foundation, a Creative Commons Attribution 4.0 Generic License has already been assigned to the Author Accepted Manuscript version that might arise from this submission. The author gratefully acknowledges the support provided by the Department of Science and Technology (DST), Government of India, through the INSPIRE Faculty Fellowship scheme.

%\clearpage



\bibliographystyle{elsarticle-num}
\bibliography{main}

%\clearpage

\appendix
\section{\sysname-enabled MESAS (Figure~\ref{fig:modular-fhe-mesas}).}\label{mesas}
MESAS is implemented under \sysname by privately computing all statistical distances from perturbed client updates and performing the clustering and pruning in the plaintext domain. Each server privately computes the following operations on encrypted, randomly masked client updates: cosine similarity, Euclidean distance, count, min-max, and variance. After collaborative decryption, the servers obtain only obfuscated metric values, which preserve relative ordering. MESAS then applies its statistical tests to identify significant deviations, followed by clustering to separate benign and malicious client updates.

\begin{figure}[!h]
\setlength{\fboxsep}{.9pt}
\begin{center}
\begin{tcolorbox}[enhanced,
   drop fuzzy shadow southwest, colframe=black,colback=white]
\footnotesize

\noindent\textbf{Input:} 
\begin{itemize}
    \item Each server $k$ receives a set of $n$ encrypted client updates $\{c_{i,k} = \Enc{r\cdot u_{i,k}}: i=1,2,\dots,n\}$, where $u_{ik}$ is $k^{th}$ subset of client update of client \textit{i}. %$f$ is the bound on the number of Byzantine clients. 
    \item Additionally receives encrypted norm and mean $\{e_{i} = \Enc{\frac{r}{\| u_{i}\|}}: i=1,2,\dots,n\}, \Enc{r\cdot\mu_i}.$ [Note - Both encrypted norm and mean can also be calculated privately at the server side.]
    \item Server Update Vector: $u_0, \|u_0\|$
\end{itemize}

\noindent\rule{\textwidth}{0.1pt}

\noindent\textbf{Server-Side Encrypted Operations}
\begin{enumerate}
\item \textit{Coordinate-wise products \textbf{[Cosine Similarity]}.}
For each client \(i\) 
\[
 \Encop_{\mathsf{pk}}(cos_i) \;=\;
  \frac{u_{0} \odot c_{i,k}}{\|u_0\| \cdot \Encop_{\mathsf{pk}}\left(\frac{r}{\|u_i\|}\right)} 
  % = \Encop_{\mathsf{pk}}(r \cdot u_{i,k} \odot u_0)
\]

\item \textit{Compute \textbf{(Euclidean distance)}}
\[
\Enc{euc_{i}} = (c_{i,k}-u_{0,k})^2.
\]
% Note that,
% \[
% \Enc{euc_{i}}=\Enc{(r \cdot u_{i,k}-r \cdot u_{0,k})^2}.
% \]

\item \textit{Find \textbf{(Pairwise distance - Min/Max/COUNT)}}
\[
\Enc{count_{i}} = (c_{i,k}-u_{0,k}).
\]

\item \textit{Calculate \textbf{(Variance)}}
\[
\Enc{var_{i}} = (c_{i,k}-r \cdot \mu_{i}).
\]

\item \textit{Collaborative decryption.}
Decrypt to obtain $cos_i$, $euc_i$, $count_{i}$, $var_{i}$.

\item []\textbf{Verification of Correct Computation.}
    % \item Server Publishes $g_0^{d_{i_l}}$ as a proof for the correct computation (for every coordinate \(\ell\)), for all  $d_i \in \{cos_i$, $euc_i$, $count_{i}$, $var_{i}\}$. 
    % \item Peer Server: For every \(\ell\) check
    %     $$h_{i,k_l}^{u_{t_l}} \stackrel{?}{=}  g^{d_{i_l}}.$$

    \item Servers publish $e(g_0,g_0)^{d_{ij,k_l}}$ as a proof for the correct computation, for all  $d_i \in \{cos_i$, $euc_i$, $count_{i}$, $var_{i}\}$.
    
    \item Peer Servers: For every coordinate \(\ell\) and pair \((i,j)\) verify,
    \[
      e(h_{ik_{\ell}}(h_{jk_{\ell}})^{-1},h_{ik_{\ell}}(h_{jk_{\ell}})^{-1}) \stackrel{?}{=} e(g_0,g_0)^{d_{ij,k_l}}.
    \]
    \item Abort on the first mismatch.

\noindent\rule{0.9\textwidth}{0.1pt}

\noindent\textbf{Plaintext operations}
\item \textbf{Euclidean and cosine distances:} Now, one of the servers (assuming at least one is honest) compute the total Euclidean and cosine distance. The distances are scaled up by $r$. However, they preserve order, which is sufficient.
\item \textbf{Min/Max/COUNT:} From all the decrypted count vectors, Min and Max can be calculated. Further, we can prepare a binary sign vector $sgn$, such that $sgn_l = \begin{cases}
    0,& \text{if } u_{i,k_l} \leq u_{0,k_l}\\
    1,              & \text{otherwise}
\end{cases}$. 
    Sum of sgn gives COUNT that how many parameter values have increased from the respective parameter of the previous global model.

    \item \textbf{Selection:} This step can be computed same as in plaintext, perform pruning based on clustering, and selecting the next global model update.
\end{enumerate}

\end{tcolorbox}
\end{center}
\caption{Instantiation of MESAS in \sysname} 
\label{fig:modular-fhe-mesas}
\end{figure}

\section{Input-Phase Confidentiality}\label{app1}
\begin{theorem}\label{thm:input-conf}
The pre-processing view \(\mathsf{view}^{\mathrm{pre}}_k\) held by any
single PPT server is computationally indistinguishable from a view
with uniformly random ciphertexts.  Consequently, the server learns
nothing about individual client updates beyond public-size metadata.
\end{theorem}

\begin{proof}
IND-CPA security (decisional Ring-LWE) makes
\((c_{i,1},c_{i,2})\) indistinguishable from encryptions of uniform
vectors.  Commitments \(h_{i,k_\ell}=g^{u_{i,k_\ell}}\) reveal no
message information under DLP hardness.  A simulator that replaces
each ciphertext by an encryption of a uniform vector produces a view
that is indistinguishable from the real execution, establishing the
claim.
\end{proof}

\section{Post-Decryption Privacy}\label{app2}

\begin{theorem}\label{thm:post-priv}
For any probabilistic polynomial-time adversary controlling a single
server, the distribution of \(\{d_i\}\) is computationally
indistinguishable from uniform; the adversary therefore gains no
non-negligible information about the underlying plaintext updates
beyond the masked vectors \(d_i\).
\end{theorem}

\begin{proof}
The global mask \(r\overset{\$}{\leftarrow}\mathbb Z_p^{\times}\)
is from a uniform distribution and remains secret from every individual
server.  Fix an arbitrary probabilistic polynomial-time adversary
\(\mathcal A\) that corrupts \(\mathcal S_1\) and let
\(\mathsf{view}^{\mathrm{post}}_1\) denote the adversary’s entire
view, including the set
\( \{d_i\}_{i=1}^{n} \) obtained after collaborative decryption.

\smallskip
\textbf{Hybrid argument.}
We build a sequence of hybrids
\(\mathcal H_0,\mathcal H_1\) and show that
\(\mathcal A\) cannot distinguish them.

\begin{itemize}
\item [\(\mathcal H_0\)] The \emph{real} execution outputting the true
      masked statistics \(d_i\) computed as specified below.

\item [\(\mathcal H_1\)]  Identical to \(\mathcal H_0\) except that the
      simulator samples a fresh
      \(r'\overset{\$}{\leftarrow}\mathbb Z_p^{\times}\)
      independent of all protocol messages and replaces
      every \(d_i\) by
      \(d_i' \gets f(r',s_i)\), where the masking function
      \(f\) is chosen per aggregation rule:
      \[
        f(r',s_i)=
        \begin{cases}
          r'^2\,s_i & \text{Krum} \\[4pt]
          r'\,s_i   & \text{Trimmed-Mean} \\[4pt]
          r'\,u_t\odot u_{i,k} & \text{FLTrust}. 
        \end{cases}
      \]
\end{itemize}

\smallskip
\textbf{Indistinguishability of hybrids.}
For each rule we argue that
\(\mathsf{view}^{\mathrm{post}}_1(\mathcal H_0)\)
and
\(\mathsf{view}^{\mathrm{post}}_1(\mathcal H_1)\)
are computationally indistinguishable:

\begin{description}

\item[Krum.]
      Coordinates take the form
      \(d_{ij,k_j}=r^{2}\,\Delta_{ij,k_j}^{2}\)
      with \(\Delta_{ij,k_j}\neq 0\) except with negligible
      probability.  Because \(r\) is uniform non-zero, \(r^{2}\)
      is uniform over the quadratic residues modulo~\(p\).
      The set of squares forms a subgroup of index~2, whence
      \(r^{2}\,\Delta^{2}\) is uniform over the same subgroup in
      both hybrids; the distributions coincide.
\item[Trimmed-Mean.]
      For any fixed non-zero coordinate
      \(s_{i,j}=u_{i,k_j}-u_{j,k_j}\),
      the product \(r\,s_{i,j}\) is uniform in \(\mathbb Z_p\)
      because \(r\) is uniform in
      \(\mathbb Z_p^{\times}\).
      When \(s_{i,j}=0\) the value is identically~0 in both
      hybrids, which occurs with probability $1/p$ and is therefore
      negligible.  Thus every coordinate of \(d_i\) is identically
      distributed in the two hybrids.

\item[FLTrust.]
      Each coordinate is
      \(d_{i,j}=r\,u_{t,j}\,u_{i,k_j}\).
      The known constant \(u_{t,j}\) is non-zero with overwhelming
      probability, so \(r\,u_{t,j}\) is again uniform over
      \(\mathbb Z_p^{\times}\).
      Multiplying by \(u_{i,k_j}\) therefore yields a uniform
      coordinate unless \(u_{i,k_j}=0\), which happens with
      negligible probability.

\end{description}

\smallskip\textbf{Simulation.}
In \(\mathcal H_1\) the simulator need not know the
plaintext updates: it draws \(r'\) and outputs independent random
vectors of the correct dimensionality.
Therefore \(\mathcal H_1\) can be generated without access to any
client secret, and
\(\mathcal A\)’s distinguishing advantage between
\(\mathcal H_0\) and \(\mathcal H_1\) is negligible.

\smallskip\textbf{Conclusion.}
Since \(\mathcal H_1\) is independent of the underlying plaintext
updates, the adversary extracts no non-negligible information from
\(\{d_i\}\).
\end{proof}

\section{Aggregation integrity}\label{app3}

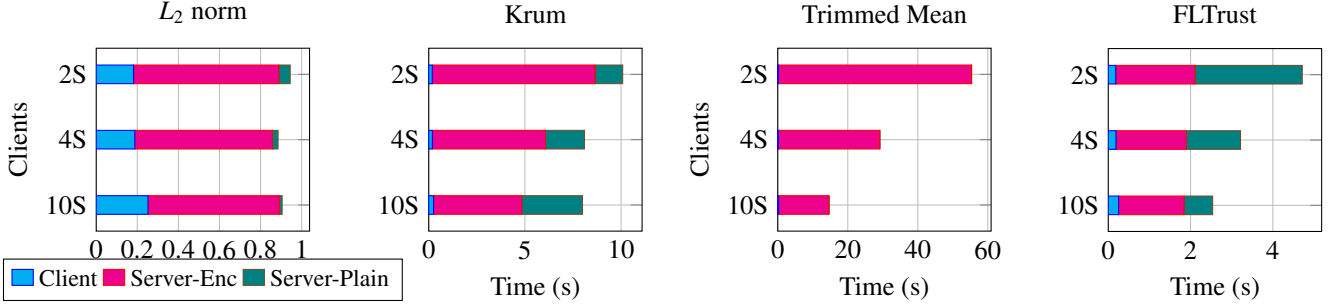
\begin{figure*}[!htbp]
\centering
\begin{subfigure}{0.24\textwidth}
\centering
\begin{tikzpicture}
\begin{axis}[
    xbar stacked,
    xmin=0,
    width=\textwidth,
    height=4cm,
    bar width=7pt,
    xlabel={Time (s)},
    ytick=data,
    yticklabels={10S, 4S, 2S},
    ylabel={Clients},
    title={$L_2$ norm},
    enlarge y limits=0.2,
    legend style={font=\small, at={(0.5,-0.15)}, anchor=north, legend columns=3},
    xtick style={draw=none},
    grid=major
]
\addplot+[xbar, fill=cyan]     coordinates {(0.252,0) (0.188,1) (0.182,2) };
\addplot+[xbar, fill=magenta]   coordinates {(0.639057,0) (0.669384,1) (0.706781,2) };
\addplot+[xbar, fill=teal]   coordinates {(0.01386,0) (0.027678,1) (0.055314,2) };
\legend{Client, Server-Enc, Server-Plain}
\end{axis}
\end{tikzpicture}
\end{subfigure}
\begin{subfigure}{0.24\textwidth}
\centering
\begin{tikzpicture}
\begin{axis}[
    xbar stacked,
    xmin=0,
    width=\textwidth,
    height=4cm,
    bar width=7pt,
    xlabel={Time (s)},
    ytick=data,
    yticklabels={10S, 4S, 2S},
    ylabel={},
    title={Krum},
    enlarge y limits=0.2,
    xtick style={draw=none},
    grid=major
]
\addplot+[xbar, fill=cyan]     coordinates {(0.252,0) (0.188,1) (0.182,2) };
\addplot+[xbar, fill=magenta]   coordinates {(4.5825,0) (5.87954,1) (8.4605,2)};
\addplot+[xbar, fill=teal]   coordinates {(3.15593,0) (2.02654,1) (1.43946,2)};
\end{axis}
\end{tikzpicture}
\end{subfigure}
\begin{subfigure}{0.24\textwidth}
\centering
\begin{tikzpicture}
\begin{axis}[
    xbar stacked,
    xmin=0,
    width=\textwidth,
    height=4cm,
    bar width=7pt,
    xlabel={Time (s)},
    ytick=data,
    yticklabels={10S, 4S, 2S},
    ylabel={Clients},
    title={Trimmed Mean},
    enlarge y limits=0.2,
    xtick style={draw=none},
    grid=major
]
\addplot+[xbar, fill=cyan]     coordinates {(0.252,0) (0.188,1) (0.182,2) };
\addplot+[xbar, fill=magenta]   coordinates {(14.3394,0) (28.9327,1) (55.12365,2) };
\addplot+[xbar, fill=teal]   coordinates {(0.0220979000000003,0) (0.0441077000000014,1) (0.0899727999999982,2)};
\end{axis}
\end{tikzpicture}
\end{subfigure}
\begin{subfigure}{0.24\textwidth}
\centering
\begin{tikzpicture}
\begin{axis}[
    xbar stacked,
    xmin=0,
    width=\textwidth,
    height=4cm,
    bar width=7pt,
    xlabel={Time (s)},
    ytick=data,
    yticklabels={10S, 4S, 2S},
    ylabel={},
    title={FLTrust},
    enlarge y limits=0.2,
    xtick style={draw=none},
    grid=major
]
\addplot+[xbar, fill=cyan]     coordinates {(0.252,0) (0.188,1) (0.182,2) };
\addplot+[xbar, fill=magenta]   coordinates {(1.58902,0) (1.70594,1) (1.92367,2)};
\addplot+[xbar, fill=teal]   coordinates {(0.696,0) (1.32137,1) (2.60926,2)};
\end{axis}
\end{tikzpicture}
\end{subfigure}

\caption{Runtime breakdown across different Byzantine-robust aggregation methods and servers (50 clients, 100k params).}
\label{fig:byzantine_time_breakdown1}
\end{figure*}

\begin{figure*}[!htbp]
\centering

% $L_2$ norm
\begin{subfigure}{0.24\textwidth}
\centering
\begin{tikzpicture}
\begin{axis}[
    xbar stacked,
    xmin=0,
    width=\textwidth,
    height=4cm,
    bar width=7pt,
    xlabel={Time (s)},
    ytick=data,
    yticklabels={10S, 4S, 2S},
    ylabel={Clients},
    title={$L_2$ norm},
    enlarge y limits=0.2,
    legend style={font=\small, at={(0.5,-0.15)}, anchor=north, legend columns=3},
    xtick style={draw=none},
    grid=major
]
\addplot+[xbar, fill=cyan] coordinates {(0.252,0) (0.188,1) (0.182,2)};
\addplot+[xbar, fill=magenta] coordinates {(1.30334,0) (1.3279,1) (1.41108,2)};
\addplot+[xbar, fill=teal] coordinates {(0.01404,0) (0.02803,1) (0.05531,2)};
\legend{Client, Server-Enc, Server-Plain}
\end{axis}
\end{tikzpicture}
\end{subfigure}
%
% Krum
\begin{subfigure}{0.24\textwidth}
\centering
\begin{tikzpicture}
\begin{axis}[
    xbar stacked,
    xmin=0,
    width=\textwidth,
    height=4cm,
    bar width=7pt,
    xlabel={Time (s)},
    ytick=data,
    yticklabels={10S, 4S, 2S},
    ylabel={},
    title={Krum},
    enlarge y limits=0.2,
    xtick style={draw=none},
    grid=major
]
\addplot+[xbar, fill=cyan] coordinates {(0.252,0) (0.188,1) (0.182,2)};
\addplot+[xbar, fill=magenta] coordinates {(20.7826,0) (26.4765,1) (36.4155,2)};
\addplot+[xbar, fill=teal] coordinates {(102.9784,0) (127.5765,1) (128.0855,2)};
\end{axis}
\end{tikzpicture}
\end{subfigure}
%
% Trimmed Mean
\begin{subfigure}{0.24\textwidth}
\centering
\begin{tikzpicture}
\begin{axis}[
    xbar stacked,
    xmin=0,
    width=\textwidth,
    height=4cm,
    bar width=7pt,
    xlabel={Time (s)},
    ytick=data,
    yticklabels={10S, 4S, 2S},
    ylabel={Clients},
    title={Trimmed Mean},
    enlarge y limits=0.2,
    xtick style={draw=none},
    grid=major
]
\addplot+[xbar, fill=cyan] coordinates {(0.252,0) (0.188,1) (0.182,2)};
\addplot+[xbar, fill=magenta] coordinates {(38.8388,0) (81.2373,1) (161.47628,2)};
\addplot+[xbar, fill=teal] coordinates {(0.0507205,0) (0.062934,1) (0.124963,2)};
\end{axis}
\end{tikzpicture}
\end{subfigure}
%
% FLTrust
\begin{subfigure}{0.24\textwidth}
\centering
\begin{tikzpicture}
\begin{axis}[
    xbar stacked,
    xmin=0,
    width=\textwidth,
    height=4cm,
    bar width=7pt,
    xlabel={Time (s)},
    ytick=data,
    yticklabels={10S, 4S, 2S},
    ylabel={},
    title={FLTrust},
    enlarge y limits=0.2,
    xtick style={draw=none},
    grid=major
]
\addplot+[xbar, fill=cyan] coordinates {(0.252,0) (0.188,1) (0.182,2)};
\addplot+[xbar, fill=magenta] coordinates {(3.19113,0) (3.34155,1) (3.62663,2)};
\addplot+[xbar, fill=teal] coordinates {(1.34383,0) (2.7569,1) (5.39441,2)};
\end{axis}
\end{tikzpicture}
\end{subfigure}

\caption{Runtime breakdown across different Byzantine-robust aggregation methods and server (100 clients, 100k params).}
\label{fig:byzantine_time_breakdown2}
\end{figure*}
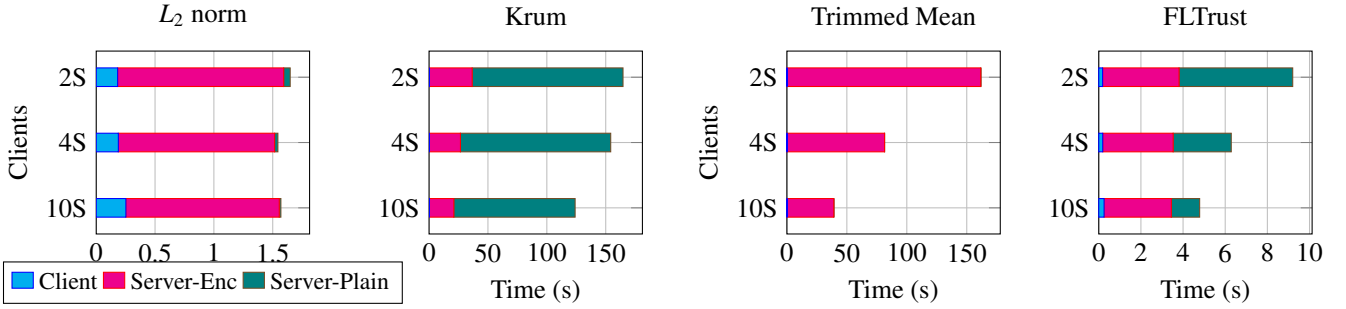

\begin{theorem}\label{thm:integrity}
Any probabilistic polynomial-time (PPT) adversary \(\mathcal A\) that
corrupts a single aggregation server,   can convince the honest
server to accept an aggregate vector \(d'_{i,j_k} \neq d_{i,j_k}\)
(where \( d_{i,j_k}\) is the actual computed intermediate aggregation vector \textsc{Krum}, Trimmed-Mean, or FLTrust as applicable)
only with negligible probability.
\end{theorem}

\begin{proof}%[Long Version]
Assume a PPT adversary~$\mathcal A$ controls server~$\mathcal S_1$ and causes the honest server~$\mathcal S_2$ to accept an incorrect vector $d'\neq d$ with non-negligible probability
$\varepsilon(\lambda)$.  We build a reduction~$\mathcal B$ breaking the discrete logarithm (DL) problem in~$\mathbb G$.

\smallskip\noindent
\textbf{Hybrid embedding.}
Given a DL instance $(g_0,g_0^{x})$, $\mathcal B$ selects a
random coordinate $\ell^{\star}$, embeds $g_0^{x}$ into the proof
element that $\mathcal S_1$ must output for that coordinate, and
generates all other protocol values honestly.  Specifically, the
expected proof in each method is replaced as follows:

\begin{center}
\resizebox{0.5\textwidth}{!}{%
\begin{tabular}{@{}|l|l|l|@{}}
\hline
Rule & Correct proof for coord.\,$\ell^{\star}$ & Challenge \\ \hline
Krum          & $e(g_0,g_0)^{d_{\ell^{\star}}}$ & $e(g_0,g_0)^{x}$ \\
Trimmed-Mean  & $g_0^{d_{\ell^{\star}}}$        & $g_0^{x}$        \\
FLTrust      & $g_0^{d_{\ell^{\star}}}$        & $g_0^{x}$        \\ \hline
\end{tabular}
}
\end{center}

\smallskip\noindent
\textbf{Forgery implies DL solution.}
If $\mathcal A$ forges a vector $d'$ that differs from $d$ in
$\ell^{\star}$ yet passes the verification test, the following
relations hold:

\[
\begin{aligned}
\text{Krum:}\;&
e(h_i h_j^{-1},h_i h_j^{-1})=e(g_0,g_0)^{x},\\
\text{Trimmed:}\;&
h_i h_j^{-1}=g_0^{x},\\
\text{FLTrust:}\;&
h_i^{u_{t,\ell^{\star}}}=g_0^{x},
\end{aligned}
\]
where $h_i=g^{u_{i,\ell^{\star}}}=g_0^{r u_{i,\ell^{\star}}}$.
Bilinearity (or plain exponent algebra) yields
\(
  x=
  \begin{cases}
    r^{2}\Delta^{2}&\text{Krum}\\
    r\Delta        &\text{Trimmed-Mean}\\
    r\,u_{t,\ell^{\star}} &\text{FLTrust},
  \end{cases}
\)
with $\Delta=u_{i,\ell^{\star}}-u_{j,\ell^{\star}}\neq0$ except with
probability $1/p$.  Because $\Delta$ and $u_{t,\ell^{\star}}$ are
non-zero and publicly known, $\mathcal B$ inverts the multiplicative
constant and outputs~$x$, solving DL whenever the forgery occurs.

\smallskip\noindent
\textbf{Extraction via the decrypted value \boldmath{$d_i$}.}
After the collaborative decryption step each intermediate statistic
\(d_{ij,k_\ell}\) is revealed to both servers in \emph{plaintext}
form.  Hence the reduction \(\mathcal B\) can read the forged entry
\(d'_{ij,k_{\ell^\star}}\) directly.  
From the verification equation that \(\mathcal A\) caused to hold we
obtain, for the special coordinate \(\ell^\star\),

\[
  x=
  \begin{cases}
    d'_{ij,k_{\ell^\star}}          &\text{Krum, since } 
         d'_{ij,k_{\ell^\star}}=r^{2}\Delta^{2},\\[4pt]
    d'_{ij,k_{\ell^\star}}          &\text{Trimmed-Mean, where } 
         d'_{ij,k_{\ell^\star}}=r\Delta,\\[4pt]
    \displaystyle
    \frac{d'_{i,\ell^\star}}{u_{t,\ell^\star}}
                                    &\text{FLTrust, with } 
         d'_{i,\ell^\star}=r\,u_{t,\ell^\star}.
  \end{cases}
\]

Here \(\Delta=u_{i,\ell^\star}-u_{j,\ell^\star}\neq0\) except with
probability \(1/p\).
Because \(r\) is uniform in \(\mathbb Z_p^{\times}\) and never
revealed, \(d'_{ij,k_{\ell^\star}}\) is a non-zero multiple of \(x\)
chosen uniformly at random, and the multiplier
(\(r^{2}\Delta^{2}\), \(r\Delta\), or \(r\,u_{t,\ell^\star}\))
is invertible with overwhelming probability.
\(\mathcal B\) therefore recovers \(x\) by a single modular
division, solving the discrete-log instance whenever the forgery
occurs.

\smallskip\noindent
\textbf{Success probability.}
The reduction fails only when the chosen coordinate is
\(\Delta=0\) (or \(u_{t,\ell^\star}=0\) for FLTrust), an event of
probability at most \(1/p\). Thus \(\mathcal B\) solves the DL with probability at least
\(\varepsilon(\lambda)/N-negl(\lambda)\), contradicting DL hardness. Consequently, a corrupt server can cause acceptance of
\(d'_{ij,k}\neq d_{ij,k}\) only with negligible probability.

\end{proof}

\section{Runtime Breakdown}\label{app5}
Figures~\ref{fig:byzantine_time_breakdown1} and~\ref{fig:byzantine_time_breakdown2} present a runtime breakdown of server (Plaintext vs. Encrypted) and client in \sysname-tailored schemes, with different numbers of server (2, 4 and 10) and client (50 and 100 clients). Each bar is divided into three components: (i) Client computation, (ii) Encrypted server-side computation, and (iii) Plaintext (or Plaintext-Ciphertext) server-side computation.

Client Computation is consistently lightweight. For all aggregation methods and client counts, the client-side computation remains at 0.18s with 2 servers. It is 0.25s for 10 servers, because the client needs to encrypt more packets. Overall, the client-side overhead does not scale significantly with the number of clients or the choice of aggregation. For 100 clients, Krum requires more than 100s on the server, indicating a drastic increase in encrypted computation due to its pairwise distance computations. For Trimmed-mean, encrypted computation time increases with the number of clients, rising from 14s (50 clients) to 161s (100 clients), because the number of pairwise comparisons increases. 
However, Trimmed-mean improves by big margins after increasing the number of servers. FLTrust shows moderate runtime overheads under encryption and remains efficient even with 100 clients.

\end{document}